\documentclass[onecolumn, a4size, 11pt]{IEEEtran}
\usepackage{amsmath}
\usepackage{amssymb}
\usepackage{amsfonts}
\usepackage{graphicx}
\usepackage{epsfig}
\usepackage{subfigure}
\usepackage{psfrag}
\usepackage{color}
\usepackage[noadjust]{cite}
\usepackage{multirow}
\usepackage{algorithm}
\usepackage{algpseudocode}
\usepackage{pifont}

\title{Distributed Wireless Power Transfer with \\ Energy Feedback 
\footnote{This paper has been presented in part at IEEE Wireless Communications and Networking Conference (WCNC), April 3-6, 2016, Doha, Qatar \cite{C_LZ:2016}.  \\
\indent The authors are with the Department of Electrical and Computer Engineering, National University of Singapore (email:  s.lee@u.nus.edu, elezhang@nus.edu.sg). R. Zhang is also with the Institute for Infocomm Research, A*STAR, Singapore. } 
}
\author{Seunghyun Lee and  Rui Zhang}

\newtheorem{lemma}{\underline{Lemma}}[section]
\newtheorem{corollary}{\underline{Corollary}}[section]
\newtheorem{proposition}{\underline{Proposition}}[section]

\def\l{\left}
\def\r{\right}
\def\({\left(}
\def\){\right)}

\def\b0{{\mathbf{0}}}

\def\cA{\mathcal{A}}
\def\cB{\mathcal{B}}

\def\cM{\mathcal{M}}

\newcommand{\nn}{\nonumber}

\begin{document}
\maketitle \thispagestyle{empty}
\begin{abstract}
Energy beamforming (EB) is a key technique for achieving efficient radio-frequency (RF) transmission enabled wireless energy transfer (WET). By optimally designing the waveforms from multiple energy transmitters (ETs) over the wireless channels, they can be constructively combined at the energy receiver (ER) to achieve an EB gain that scales with the number of ETs. However, the optimal design of EB waveforms requires accurate channel state information (CSI) at the ETs, which is challenging to obtain practically, especially in a distributed system with ETs at separate locations. In this paper, we study practical and efficient channel training methods to achieve optimal EB in a distributed WET system. We propose two protocols with and without centralized coordination, respectively, where distributed ETs either sequentially or in parallel adapt their transmit phases based on a low-complexity energy feedback from the ER. The energy feedback only depends on the received power level at the ER, where each feedback indicates one particular transmit phase that results in the maximum harvested power over a set of previously used phases. Simulation results show that the two proposed training protocols converge very fast in practical WET systems even with a large number of distributed ETs, while the protocol with sequential ET phase adaptation is also analytically shown to converge to the optimal EB design with perfect CSI by increasing the training time. Numerical results are also provided to evaluate the performance of the proposed distributed EB and training designs as compared to other benchmark schemes.
\end{abstract}


\setlength{\baselineskip}{1.3\baselineskip}

\section{Introduction}
Radio-frequency (RF) transmission enabled wireless energy transfer (WET) is a promising technology to achieve perpetual operation of wireless devices by supplying energy continually over the air. It has been reported that RF-based WET can already deliver tens of microwatts power to wireless devices from a distance of more than 10 meters.\footnote{Please refer to the company website of Powercast Corp. (http://www.powercastco.com) for more information on RF-based WET.} Motivated by this, wireless powered communication has recently emerged as a new area of research in wireless communication, where communication devices with typically low-power consumptions such as RF identification (RFID) tags and sensors distributed in a wide area are charged via WET. In particular, two appealing lines of research are  \emph{wireless powered communication network} (WPCN) \cite{J_JZ:2014,J_LZC:2014} and \emph{simultaneous wireless information and power transfer} (SWIPT) \cite{J_ZH:2013,J_XLZ:2014}, which have spurred significant research efforts (see, e.g., \cite{J_BHZ:2015, J_HZ:2015, J_DZNPSSV:2015, J_LWNKH:2015, J_BZZ:2016, J_HZZ:2016, J_KTNZNS:2014} and the references therein). In WPCN, wireless devices are powered by dedicated downlink WET for their uplink wireless information transmission (WIT); while in SWIPT, a dual use of RF signals is considered for simultaneous downlink WET and WIT. In both WPCN and SWIPT, efficient design of WET to compensate the significant power loss of RF signals over long distance is essential.  

\begin{figure}
\centering
\subfigure[In-band WET and WIT]{
\centering
\includegraphics[width=7cm]{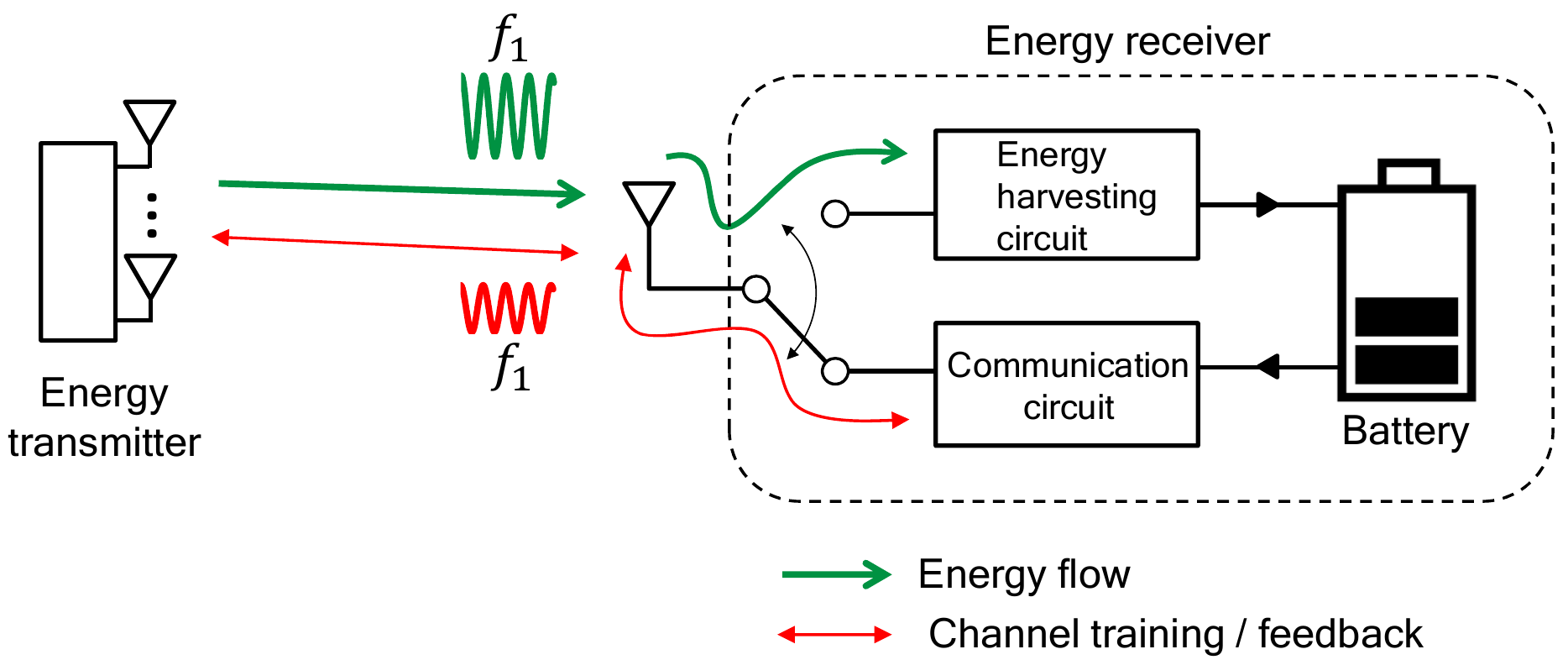}\label{Fig:InBand}} 
\subfigure[Out-band WET and WIT]{
\centering
\includegraphics[width=7cm]{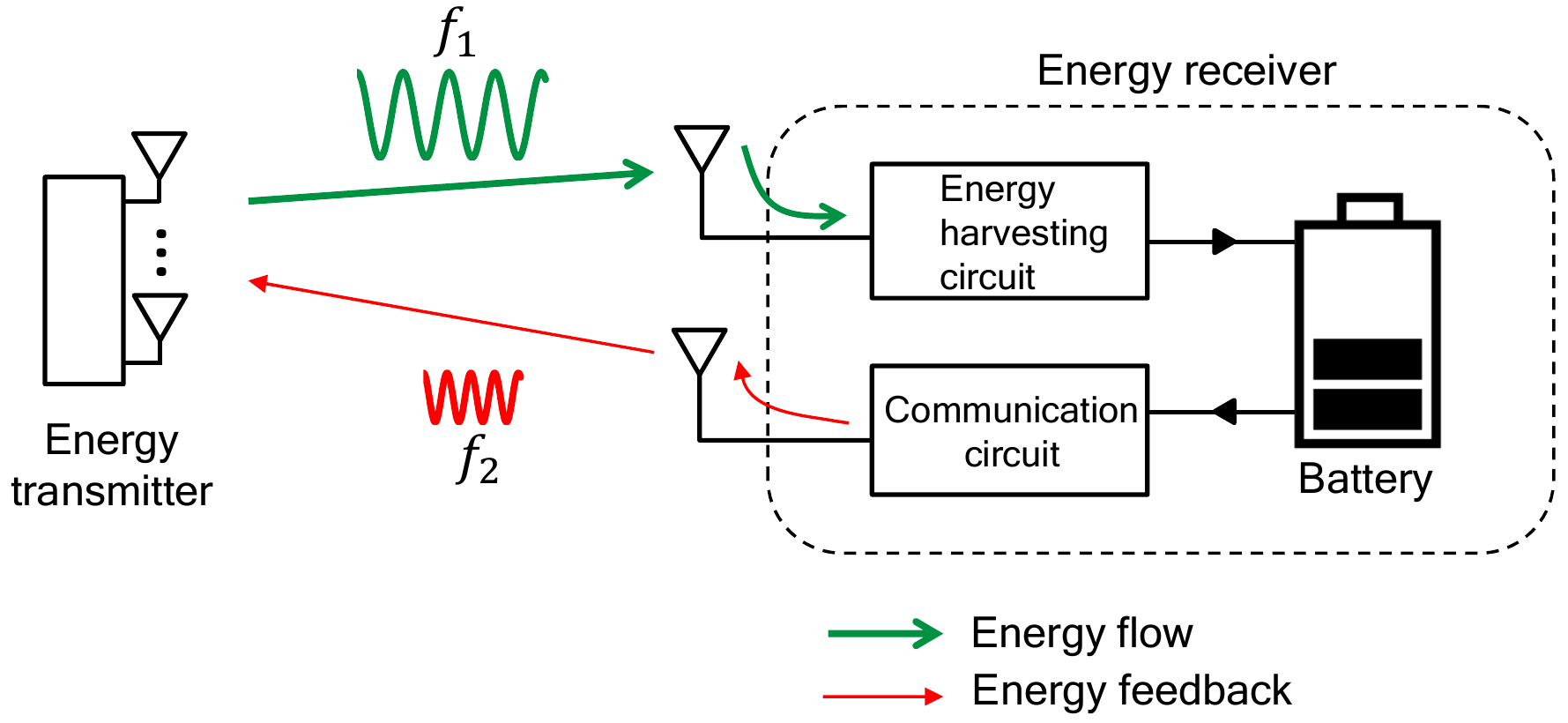}\label{Fig:OutBand}} 
\caption{Comparison of in-band versus out-band WET and WIT.}
\label{Fig:ER}
\end{figure}

In WET systems, \emph{energy beamforming} (EB) is a key technique used to significantly enhance the WET efficiency. With EB, the signal waveforms from multiple transmit antennas are optimally designed such that after propagating through different wireless channels they are constructively combined at a destined energy receiver (ER) to maximize the received signal amplitude and thus average power. However, in practice the energy transmitters (ETs)  need to acquire accurate knowledge of the channel state information (CSI) to the ER to achieve optimal EB. To this end, various CSI acquisition methods have been proposed in the literature depending on the type of ER model assumed. In general, two practical ER models have been considered for \emph{in-band} and \emph{out-band} WET and WIT, respectively, as shown in Fig.~\ref{Fig:ER}. 

For in-band WET, as shown in Fig.~\ref{Fig:InBand}, the WET and WIT are assumed to be implemented in the same frequency band, and  a single antenna is used at the ER for both energy harvesting and communication in a time-division-duplexing (TDD) manner. In this case, similar to the channel estimation techniques in conventional wireless communication \cite{J_LHLGRA:2008}, the ET can send pilot signals to the ER that uses the communication circuit to estimate the channel and then send back the channel estimation to the ET for implementing EB \cite{J_YHG:2014}. However, this forward training approach incurs significant training and feedback overhead as the number of antennas at the ET increases. To overcome this difficulty, an alternative approach of reverse training is proposed in \cite{J_ZZ:2015, J_ZZ:2015_b}, where  training signals are sent by the ER to the ET to estimate the reverse-link channel that is assumed to be reciprocal of the forward-link channel over which the EB is implemented; as a result, the training overhead is independent of the number of antennas at the ET. However, it should be noted that unlike the conventional reverse-link training based channel estimation in wireless communication, a new design metric called \emph{net energy maximization} needs to be considered in the case of WET, where the net energy refers to the average harvested energy minus that used for the reverse link training at the ER, so as to optimally balance the trade-off between minimizing its energy consumption for sending pilot signals and maximizing its harvested energy with a larger EB gain due to more accurately estimated channels at the ET (which, however, requires more training energy used).

On the other hand, for the case of out-band WET where the WET and WIT are implemented over different frequency bands, as shown in Fig.~\ref{Fig:OutBand}, two antennas are used at the ER for energy harvesting and communication over two orthogonal frequencies, respectively. Note that in the in-band WET, the front end of the ER needs to be re-designed to switch between energy harvesting and communication; while the ER in the out-band WET can be more easily implemented with the off-the-shelf antenna and rectifier (integrated as a so-called rectenna) for RF energy harvesting and a separate antenna for communication. However, unlike in-band WET, in out-band WET the ET cannot obtain the channel knowledge to the energy harvesting antenna at the ER by conventional forward/reverse training methods due to the following two reasons. Firstly, the communication antenna at the ER operates at a different frequency from that of the energy harvesting antenna. Secondly, practical rectifiers for RF energy harvesting do not have the baseband processing capability required for channel training or estimation. Therefore, a new channel estimation method based on the feedback only pertaining to the measured power level at the energy harvesting antenna of the ET, namely \emph{energy feedback}, is proposed in \cite{J_XZ:2014, J_XZ:2016}. In this method, the energy feedback from the ER is used for iteratively refining the estimate of the multiple-input multiple-output (MIMO) channel from a multi-antenna ET to multiple ER antennas by applying the cutting-plane method in convex optimization. 

\begin{figure} 
\centering
\includegraphics[width=10cm]{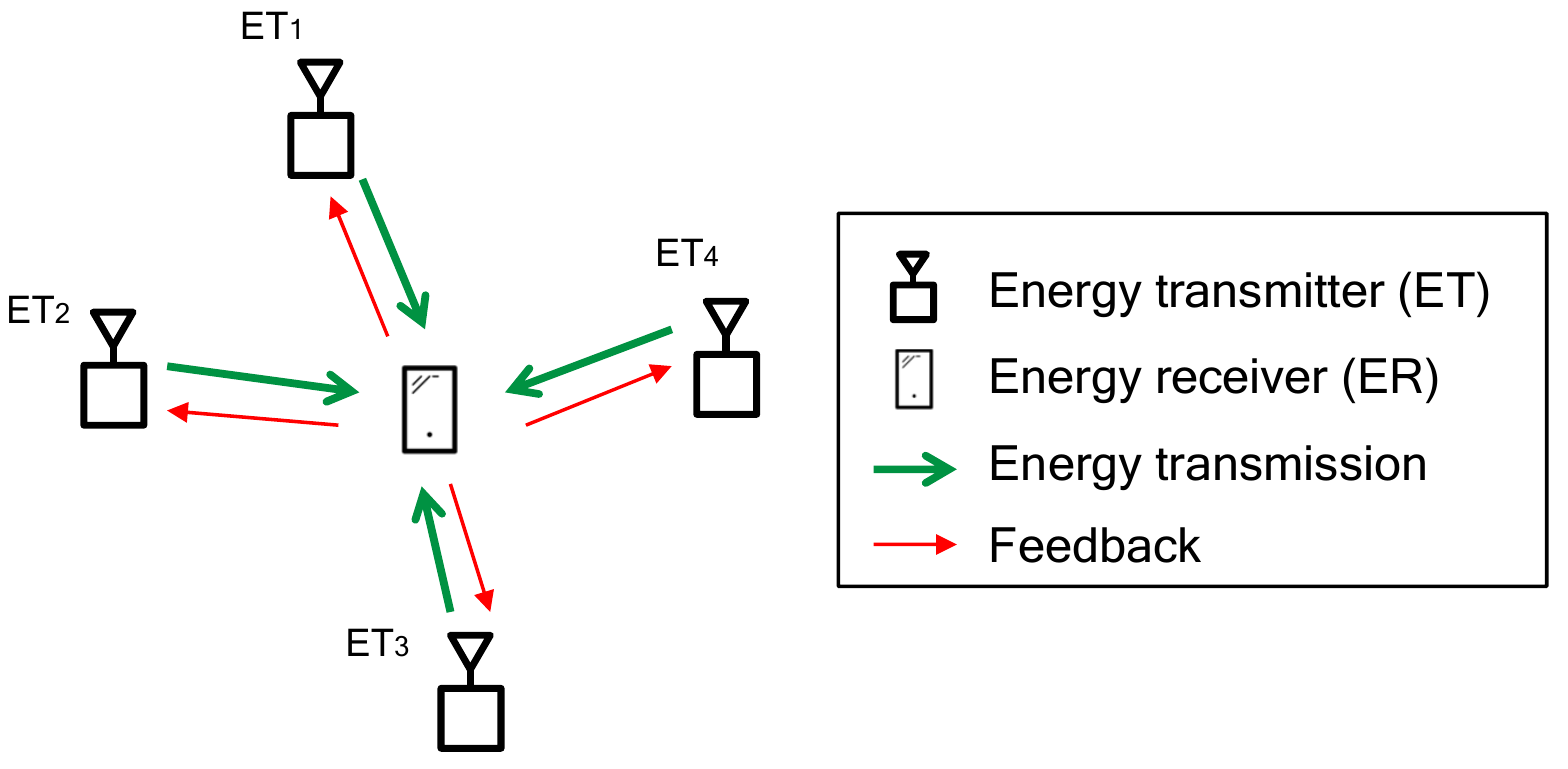}
\caption{A WET system with distributed EB based on the ER feedback.} \label{Fig:SystemModel}
\end{figure}

In this paper, we consider the optimal EB design in a distributed WET system with multiple separated single-antenna ETs that cooperatively send power to one single-antenna ER assuming the out-band WET model, as illustrated in Fig.~\ref{Fig:SystemModel} for the case of four ETs. Note that unlike \cite{J_XZ:2014, J_XZ:2016}, where all transmit antennas are equipped at one single ET and thus their channels to the ER can be jointly estimated, the distributed EB considered in this paper needs to be implemented over separate ETs without the need of their centralized processing. In practice, EB via distributed ETs each with isotropic transmission has the advantage of avoiding high power intensity from any single ET to the ER as compared to the conventional EB by a single multi-antenna ET as considered in \cite{J_YHG:2014, J_ZZ:2015, J_ZZ:2015_b, J_XZ:2014, J_XZ:2016}, thus significantly improving its safety in practical operation. For example, distributed antennas for SWIPT assuming perfect CSI at the transmitters are considered in \cite{J_LLZ:2015}. Motivated by the above discussions, in this paper we focus on studying new channel training methods for distributed EB, based on the practical energy feedback from the ER. It is worth noting that our proposed distributed channel training and beamforming schemes can also be applied to other scenarios in wireless communication based on energy feedback from the receiver, such as cognitive radio networks as studied in \cite{J_NG:2013}, \cite{J_GS:2015} when the transmit antennas are not co-located.   

The main contributions of this paper are summarized as follows.
\begin{itemize}
 \item First, we propose two practical channel training methods for implementing distributed EB based on the ER's energy feedback, namely \emph{sequential training} and \emph{parallel training} with and without the centralized coordination among the ETs, respectively. In sequential training, the ETs take turns to adapt their transmit phases over different phase adaptation intervals with the phases of all other non-adapting ETs being fixed. As a result, the transmit phases of all ETs are optimally aligned in a sequential manner to maximize the EB gain at the ER. In contrast, in the parallel training case, each ET independently decides on whether adapting the transmit phase at each interval by flipping a coin with a certain  probability. Consequently, a random subset of ETs simultaneously adjust their transmit phases to match the phase of the received sum-signal from all other non-adapting ETs transmitting with fixed phases, to maximize the received energy at the ER with best effort. 
 
\item The above two training designs are both based on a new transmit phase adaptation algorithm proposed, where each adapting ET iteratively tunes the phase of its transmit signal to optimally align to the sum-signal from all other non-adapting ETs with fixed transmit phases at the ER. Specifically, for a given integer $B\geq 1$, each adapting ET transmits with a sequence of $2^B$ preassigned phases; then the ER feeds back $B$ bits to the adapting ET indicating one of the $2^B$ phases that results in the largest harvested power, based on which each adapting ET is able to iteratively localize the target phase to maximize the energy received by the ER. This phase adaptation algorithm is also extended to the case where the ER has a memory that stores the value of the maximum harvested power up to each training time. It is shown that by exploiting this information, the adapting ET can save the number of transmit phases in each training compared to the case without memory, thus further improving the training efficiency.

\item Next, we analyze the convergence performance of the proposed schemes. First, it is shown that the proposed phase adaptation algorithms (with and without memory) can both converge to the target phase exponentially fast with increasing training time in each phase adaptation interval. It is also proved that $B=1$, i.e., one bit per ER feedback for the algorithm without memory and two bits per ER feedback for the algorithm with memory, is optimal in terms of estimated transmit phase accuracy given the same training time. In addition, under the same training time and value of $B$, the algorithm with memory always performs better than that without memory. Furthermore, the proposed distributed EB protocol with sequential training is proved to converge to the optimal EB solution with perfect CSI as the training time increases.

\item Finally, we provide extensive simulation results under various setups to compare the training and beamforming performance of our proposed schemes in distributed WET systems, as compared to other heuristic designs. It is shown by simulation that the convergence speed of the proposed sequential training is in general faster than that of parallel training, at the cost of additional scheduling coordination among the ETs.

\end{itemize}

The rest of this paper is organized as follows. Section~\ref{Section:SystemModel} introduces the system model for distributed WET. Section~\ref{Section:Protocol} presents two proposed distributed EB protocols with sequential training and parallel training, respectively. Section~\ref{Section:PhaseAdaptation} presents the key phase adaptation algorithm used in both protocols, with or without the ER memory. Section~\ref{Section:Performance} presents the analytical results on the performance of the proposed schemes. Section~\ref{Section:Simulation} provides simulation results. Finally, Section~\ref{Section:Conclusion} concludes the paper.

\section{System Model} \label{Section:SystemModel}
As shown in Fig.~\ref{Fig:SystemModel}, we consider a distributed WET system where $M>1$ single-antenna ETs collaboratively send wireless power to a single-antenna ER. For the ER, we adopt the out-band WET and WIT model in Fig.~\ref{Fig:OutBand}, and assume that it can send energy feedback to all ETs at a given frequency different from that for the WET.

Since energy signals carry no information, for simplicity we assume the transmit signal of ET$_m$, $m=1,...,M$, to be an unmodulated carrier signal with phase offset $\phi_m\in [-\pi, \pi)$, which is expressed as
 \begin{equation} \label{Eq:TxSignal1}
s_m(t) =  \sqrt{2P}\cos(2\pi f_c t + \phi_m), 
\end{equation}
where $P$ denotes the transmit power of each ET, and $f_c$ is the carrier frequency. Each transmitted signal propagates through a  multi-path wireless channel in general. Thus, the received signal $r(t)$ at  ER from all ETs is expressed as 
\begin{align} \label{Eq:RxSignal1}
r(t) = \sum_{m=1}^M \sum_{l=1}^{L_m} a_{m,l} \sqrt{2P}  \cos(2\pi f_c (t - \tau_{m,l}) + \phi_m),
\end{align}
where $L_m$ is the number of signal paths from ET$_m$ to ER, and $a_{m,l}$, $\tau_{m,l}$ are the signal attenuation and delay of the $l$th path, respectively, with $l=1,...,L_m$. The received signal given in \eqref{Eq:RxSignal1} can be simplified as
\begin{equation} \label{Eq:RxSignal2}
r(t) = \sqrt{2P}\sum_{m=1}^M \sqrt{\beta_m} \cos(2\pi f_c t + \phi_m - \theta_m),
\end{equation}
where $\beta_m$ and $\theta_m$ are the aggregate power gain and phase shift of the multi-path channel from ET$_m$ to ER, respectively, given by
\begin{equation*}
\beta_m = \l(\sum_{l=1}^{L_m} a_{m,l}\cos(2\pi f_c \tau_{m,l})\r)^2 + \l(\sum_{l=1}^{L_m} a_{m,l}\sin(2\pi f_c \tau_{m,l})\r)^2, 
\end{equation*}
and 
\begin{equation*}
\theta_m = \arctan\l(\frac{\sum_{l=1}^{L_m} a_{m,l}\sin(2\pi f_c \tau_{m,l})}{\sum_{l=1}^{L_m} a_{m,l}\cos(2\pi f_c \tau_{m,l})}\r).
\end{equation*}
The average harvested power at ER, denoted by $Q$, is then given by
\begin{align}
Q &=  \frac{\rho}{T}\int_{0}^{T} |r(t)|^2 dt \nn\\
& = \rho P \bigg(\sum_{m=1}^M \beta_m  
 + \sum_{i,j=1, i\neq j}^M  \sqrt{\beta_i \beta_j} \cos((\phi_i - \theta_i) - (\phi_j - \theta_j))     \bigg), \label{Eq:P_r}
\end{align}
where $0<\rho \leq 1$ is the energy conversion efficiency at ER, and since it is a constant we assume $\rho=1$ in the sequel for notational convenience; and $T = \frac{1}{f_c}$ is the period of the carrier signal.

If each ET$_m$ perfectly knows the phase shift of its channel $\theta_m$, the optimal transmit phase that maximizes the harvested power $Q$ in \eqref{Eq:P_r} is given by $\phi_m^\star = \theta_m + c$, $m=1,...,M$, where $c$ is an arbitrary constant. We refer to this case as the \emph{optimal EB} in the sequel. With the optimal EB, the maximum harvested power at ER, denoted by $Q^\star$, is thus given by
\begin{equation} \label{Eq:P_r max}
Q^\star = P\l(\sum_{m=1}^M \beta_m + \sum_{i,j=1, i\neq j}^M\sqrt{\beta_i \beta_j}     \r).
\end{equation}
In practice, only imperfect CSI is available at each ET, and thus the maximum harvested power in \eqref{Eq:P_r max} with the optimal EB only provides a performance upper bound for practical distributed WET systems. In this paper, we propose new training designs of low complexity for practical distributed WET systems to maximize the EB gain and approach that of the optimal EB. Generally speaking, in our proposed training schemes the ETs in a distributed WET system adjust their transmit phases in a distributed manner, based on the energy feedback from ER. After all $M$ ETs set their transmit phases to be $\phi_m = \bar{\phi}_m$, $m=1,...,M$, via the proposed training scheme, the harvested power at ER, denoted by $Q_\text{d}$, is given by (by substituting $\phi_m = \bar{\phi}_m$, $m=1,...,M$ into \eqref{Eq:P_r})
\begin{equation} \label{Eq:P_r Protocol}
Q_\text{d} = P \bigg(\sum_{m=1}^M \beta_m  
 + \sum_{i,j=1, i\neq j}^M  \sqrt{\beta_i \beta_j} \cos((\bar{\phi}_i - \theta_i) - (\bar{\phi}_j - \theta_j))     \bigg). 
\end{equation}
Note that $Q_\text{d} \leq Q^\star$ in general according to \eqref{Eq:P_r max}. 
We define the resulting efficiency of the distributed WET system, denoted by $\eta$, as the ratio between $Q_\text{d}$ and $Q^\star$ given in \eqref{Eq:P_r Protocol} and \eqref{Eq:P_r max}, respectively, i.e.,
\begin{equation} \label{Eq:Efficiency}
\eta = \frac{Q_\text{d}}{Q^\star},
\end{equation}
where $0<\eta\leq 1$ since $Q_\text{d}\leq Q^\star$.

\section{Distributed EB Protocols} \label{Section:Protocol}
In this section, we present two distributed EB protocols for all $M$ ETs to collaboratively send power to the ER with \emph{sequential training} and \emph{parallel training}, respectively. Both protocols are based on iterative phase adaptation of ETs, while they differ in the number of ETs that adapt the transmit phase at each time. The total training time is divided into equal  \emph{phase adaptation intervals} in each of which one (in the case of sequential training) or more (in the case of parallel training) ETs may adjust their transmit phases based on the ER energy feedback, referred to as ``adapting ETs''. The adapting ETs align their signals to the sum-signal of other non-adapting ETs with fixed transmit phases at the ER to maximize the total harvested power, via a new transmit phase adaptation algorithm (to be presented later in Section~\ref{Section:PhaseAdaptation}).  The details of sequential- and parallel-training based distributed EB protocols are provided in the following two subsections, respectively.

\subsection{Distributed EB Protocol with Sequential Training} \label{Subsection:Sequential}

\begin{figure} 
\centering
\includegraphics[width=12cm]{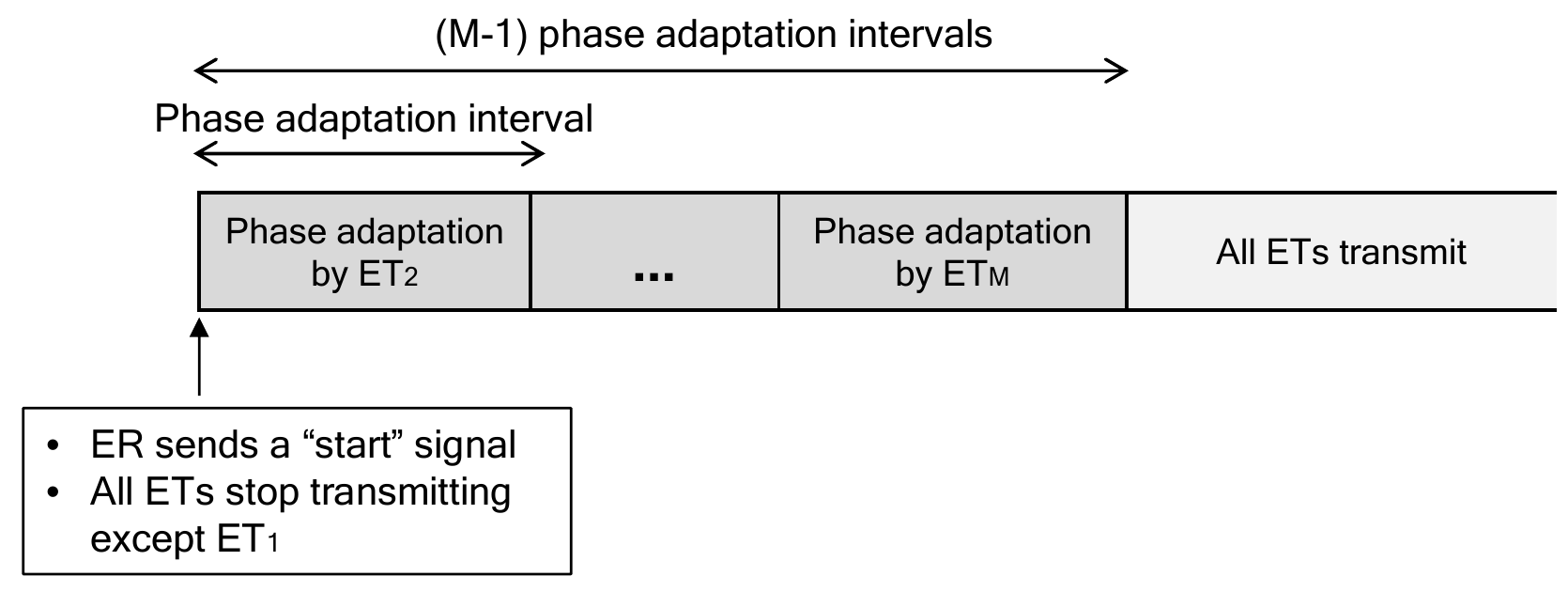}
\caption{The distributed EB protocol with sequential training.} \label{Fig:Protocol_Sequential}
\end{figure}

In this subsection, we present the distributed EB protocol with sequential training, which is illustrated in Fig.~\ref{Fig:Protocol_Sequential} and explained as follows. 
\begin{itemize}
\item To initiate distributed EB, the ER sends a ``start" signal to all $M$ ETs.\footnote{In practice, ER may send the ``start" signal when the harvested power becomes lower than some predefined threshold due to channel variation. }

\item Once the ETs receive the ``start" signal, they stop transmitting, except ET$_1$ which transmits with an arbitrary fixed phase $\bar{\phi}_1$. Without loss of generality, we assume $\bar{\phi}_1 = 0$. 

\item ET$_2$ adapts its phase to optimally align its transmitted signal to  ET$_1$'s signal at ER. After this interval, it continues to transmit with its adapted phase $\bar{\phi}_2$.

\item ET$_3$ adapts its phase to optimally align its signal to ET$_1$'s and ET$_2$'s sum-signal at ER. After this interval, it continues to transmit with its adapted phase $\bar{\phi}_3$.

\item The above procedure is repeated until ET$_M$ adapts its transmit phase; after that all $M$ ETs transmit with their adapted phases.
\end{itemize}

Note that in the above protocol with sequential training, only one ET adapts its transmit phase at each phase adaptation interval, by following a preassigned order. As a result, the transmit signals of different ETs are sequentially aligned to ET$_1$'s signal at ER, to achieve the optimal EB in \eqref{Eq:P_r max} using $(M-1)$ phase adaptation intervals in total. Notice that the sequential training requires scheduling coordination among the ETs, and can be practically implemented via e.g., the token-ring based protocol, where ETs are connected in a wired/wireless ring network and each adapting ET receives and transmits a token (which is a given sequence of bits) from/to another ET next to it in the ring at the beginning/end of each phase adaptation interval. The convergence performance of the sequential training will be analyzed in Section~\ref{Section:Performance}.

\subsection{Distributed EB Protocol with Parallel Training} \label{Subsection:parallel}

\begin{figure} 
\centering
\includegraphics[width=12cm]{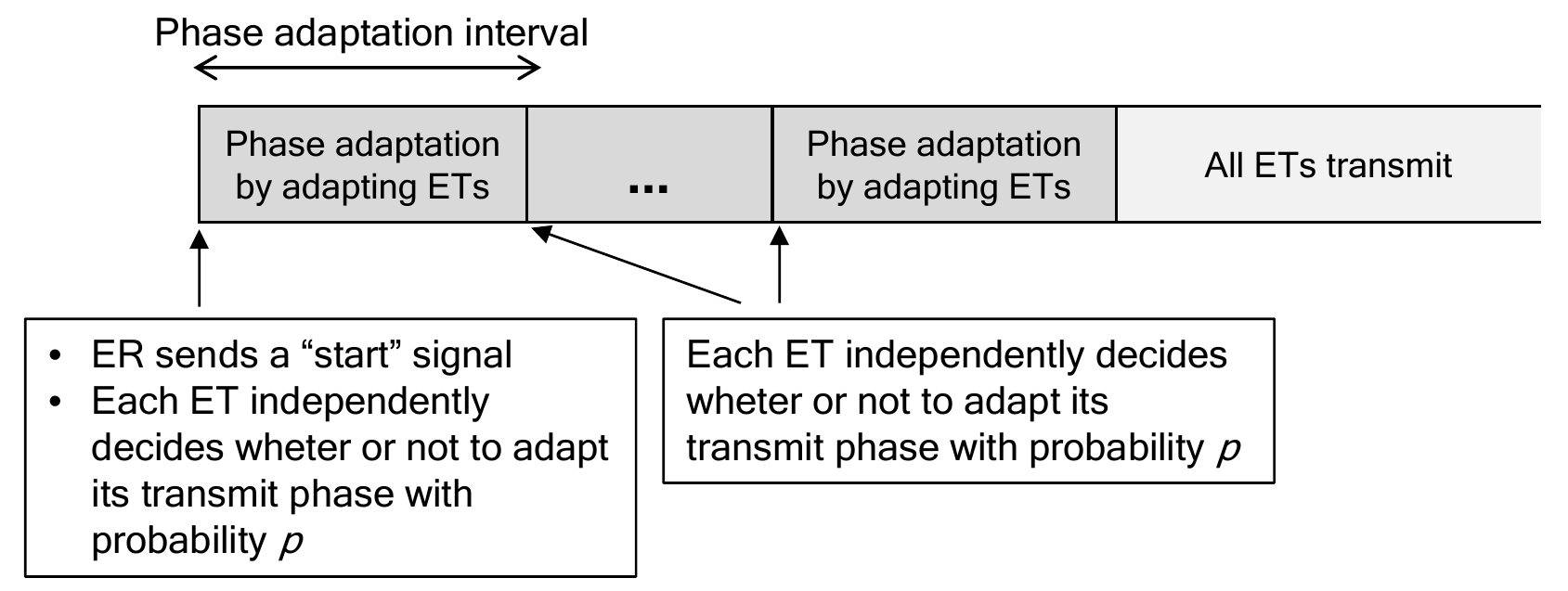}
\caption{The distributed EB protocol with parallel training.} \label{Fig:Protocol_parallel}
\end{figure}

In this subsection, we present another distributed EB protocol with parallel training, which is illustrated in Fig.~\ref{Fig:Protocol_parallel} and explained as follows.
\begin{itemize}
\item To initiate distributed EB, ER sends a ``start" signal to all $M$ ETs.

\item After all ETs receive the ``start" signal, each ET independently  decides whether or not to adapt its transmit phase, based on the outcome of randomly flipping a coin with a biased probability $p$, where $p$ denotes the probability of adapting the phase, with $0<p<1$. The adapting ETs then  simultaneously adjust their phases to optimally align their signals to the sum-signal of all other non-adapting ETs (each transmits with its phase unchanged) at ER. 

\item The above procedure is repeated for a given number of phase adaptation intervals; then all $M$ ETs transmit with their adapted phases.
\end{itemize}

Note that in the above parallel training, it may happen at each phase adaptation interval that all the ETs decide to adapt or not to adapt while in both cases the phase adaptation is not effective. However, such events are of very low probability in practice with sufficiently large number of $M$ and by setting $p$ around $0.5$, since the probability that all ETs adapt is $p^M$ and that for all not to adapt is $(1-p)^M$.

The above parallel training and the sequential training presented in Section~\ref{Subsection:Sequential} are compared in more detail as follows. Unlike the sequential training where ETs take turns for phase adaptation by following a preassigned order, in the parallel training case, a subset of adapting ETs is randomly selected at each phase adaptation interval for simultaneous phase adaptation. As a result, if the training time per interval is sufficiently large, the phase of the adapting ET will converge to that of the optimal EB (subject to a constant phase shift for all ETs) in the case of sequential training; while in the case of parallel training, the phases of all adapting ETs will converge to the same value that optimally aligns their sum-signal to that of non-adapting ETs whose transmit phases, however, may not be optimal. Thus, the parallel training needs to be repeated over intervals with randomized subsets of adapting ETs, to iteratively converge to the optimal EB with best effort. However, the number of phase adaptation intervals required can be arbitrary whereas the sequential training requires only $(M-1)$ phase adaptation intervals to achieve the optimal EB (see Fig.~\ref{Fig:Protocol_Sequential}). Nonetheless, it is also noted that the parallel training does not require the additional scheduling coordination among ETs as in the sequential training case, thus reducing the complexity for implementation. The convergence performance of the parallel training will be evaluated via simulation in Section~\ref{Section:Simulation}.

\section{Phase Adaptation Algorithm} \label{Section:PhaseAdaptation}

\begin{figure} 
\centering
\includegraphics[width=10cm]{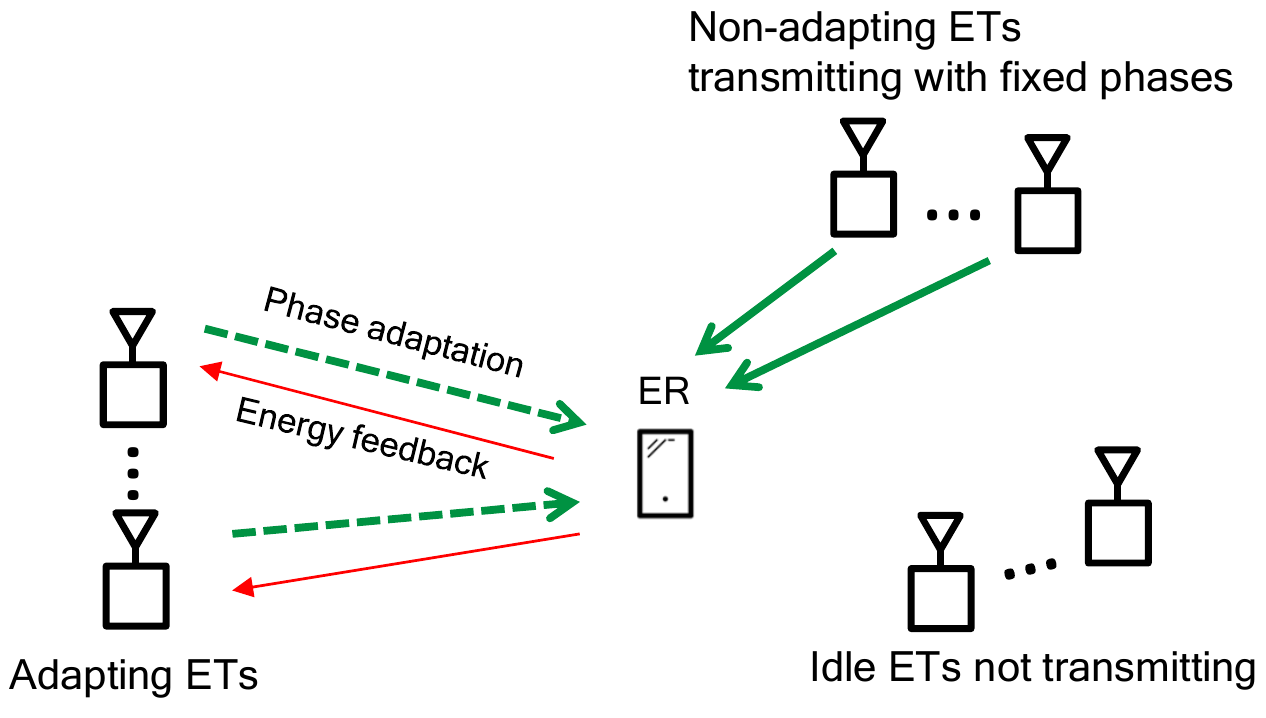}
\caption{Illustration of transmit phase adaptation by adapting ETs based on energy feedback from ER.} \label{Fig:Adapt}
\end{figure}

In this section, we present a new phase adaptation algorithm for the adapting ETs in both the sequential and parallel training protocols proposed in Section~\ref{Section:Protocol}. As illustrated in Fig.~\ref{Fig:Adapt}, in each phase adaptation interval, all ETs can be divided into three groups without loss of generality, namely \emph{adapting ETs} that adapt their transmit phases, \emph{non-adapting ETs} that do not adapt but transmit with fixed phases, and \emph{idle ETs} that neither adapt nor transmit. For convenience, we denote the set of adapting ETs, non-adapting ETs, and idle ETs as $\cM_\text{A}$, $\cM_\text{NA}$, and $\cM_\text{I}$, respectively, which are mutually exclusive and satisfy that $\cM_\text{A}\cup \cM_\text{NA} \cup \cM_\text{I} = \{1,2,...,M\}$. For sequential training, at the $m$th phase adaptation interval ($1\leq m \leq M-1$), it follows that $\cM_\text{A} = \{m+1\}$, $\cM_\text{NA} = \{1,...,m\}$, and $\cM_\text{I} = \{m+2,...,M\}$; whereas for parallel training, $\cM_\text{A}$ and $\cM_\text{NA}$ are randomly generated in each phase adaptation interval, with $\cM_\text{I} = \emptyset$. 

Let the adapting ETs transmit with a common phase $\psi$ in each phase adaptation interval, i.e., $\phi_m = \psi$, $m\in\cM_\text{A}$, and the non-adapting ETs transmit with arbitrary fixed phases, i.e., $\phi_m = \bar{\phi}_m$, $m\in\cM_\text{NA}$. For both the cases of sequential and parallel training, from \eqref{Eq:P_r} we can re-express the received signal at ER from all transmitting ETs (including both adapting and non-adapting ETs) as 
\begin{equation} \label{Eq:RxSignal4}
r(t) = \sqrt{2P}\l(\sum_{m\in\cM_\text{A}} \sqrt{\beta_m} \cos(2\pi f_c t + \psi - \theta_m) + \sum_{m\in\cM_\text{NA}} \sqrt{\beta_m} \cos(2\pi f_c t + \bar{\phi}_m - \theta_m)\r).
\end{equation}
It can then be shown that the average harvested power at the ER can be expressed as a function of $\psi$ as 
\begin{align}
Q(\psi) &=  \frac{1}{T}\int_{0}^{T} |r(t)|^2 dt \nn\\
& =  P \l(\alpha_\text{A} + \alpha_\text{NA} + 2\sqrt{\alpha_\text{A} \alpha_\text{NA}} \cos(\varphi_\text{A}(\psi) + \varphi_\text{NA})\r), \label{Eq:P_r 2}
\end{align}
where
\begin{equation*}
\alpha_\text{A} \triangleq \l(\sum_{m\in\cM_\text{A}}\sqrt{\beta_m}\cos(\psi - \theta_m) \r)^2 + \l(\sum_{m\in\cM_\text{A}}\sqrt{\beta_m}\sin(\psi - \theta_m) \r)^2,
\end{equation*}
\begin{equation*}
\alpha_\text{NA} \triangleq \l(\sum_{m\in\cM_\text{NA}}\sqrt{\beta_m}\cos(\bar{\phi}_m - \theta_m) \r)^2 + \l(\sum_{m\in\cM_\text{NA}}\sqrt{\beta_m}\sin(\bar{\phi}_m - \theta_m) \r)^2,
\end{equation*}
and 
\begin{equation*}
\varphi_\text{A}(\psi) \triangleq \arctan\l(\frac{\sum_{m\in\cM_\text{A}}\sqrt{\beta_m}\sin(\psi - \theta_m)}{\sum_{m\in\cM_\text{A}}\sqrt{\beta_m}\cos(\psi - \theta_m)} \r),
\end{equation*}
\begin{equation*}
\varphi_\text{NA} \triangleq \arctan\l(\frac{\sum_{m\in\cM_\text{NA}}\sqrt{\beta_m}\sin(\bar{\phi}_m - \theta_m)}{\sum_{m\in\cM_\text{NA}}\sqrt{\beta_m}\cos(\bar{\phi}_m - \theta_m)} \r).
\end{equation*}
Since $\alpha_\text{A}$, $\alpha_\text{NA}$, and $\varphi_\text{NA}$ can be shown to all be invariant over $\psi$, the harvested power in \eqref{Eq:P_r 2} is maximized over $\psi$ if $\varphi_\text{A}(\psi) = -\varphi_\text{NA}$. Let $\psi^*$ denote the optimal common phase of the adapting ETs to maximize $Q(\psi)$, i.e., $\varphi_\text{A}(\psi^*) = -\varphi_\text{NA}$. It can be shown that $\psi^*$ is given by
\begin{equation} \label{Eq:TargetPhase}
\psi^* = \arctan\l(\frac{(1-\tan(\varphi_\text{NA}))\sum_{m\in\cM_\text{A}}\sqrt{\beta_m}\cos(\theta_m)}{(1+\tan(\varphi_\text{NA}))\sum_{m\in\cM_\text{A}}\sqrt{\beta_m}\sin(\theta_m)}\r).
\end{equation}

The condition $\varphi_\text{A}(\psi^*) = -\varphi_\text{NA}$ to maximize $Q(\psi)$ over $\psi$ has the following intuitive explanation. The received sum-signal at the ER, i.e., $r(t)$ given in \eqref{Eq:RxSignal4}, consists of two aggregated waveforms, where one is the sum-signal from all ETs in $\cM_\text{A}$ with amplitude $\sqrt{\alpha_\text{A}}$ and phase $\varphi_\text{A}(\psi)$, and the other is the sum-signal from all ETs in $\cM_\text{NA}$ with amplitude $\sqrt{\alpha_\text{NA}}$ and phase $-\varphi_\text{NA}$. Thus, the condition $\varphi_\text{A}(\psi^*) = -\varphi_\text{NA}$ implies that the sum-signal from all adapting ETs with the best common phase $\psi^*$ is optimally aligned with that from all non-adapting ETs, which thus maximizes $Q(\psi)$ over $\psi$.

Next, we present the details of the phase adaptation algorithm for each adapting ET (applicable to both sequential and parallel trainings) for the two cases for whether the ER stores the largest received power value up to the current training time (with memory) or not (without  memory), respectively.

\subsection{Phase Adaptation Algorithm Without Memory} \label{Subsection:A1}

\begin{figure} 
\centering
\includegraphics[width=12cm]{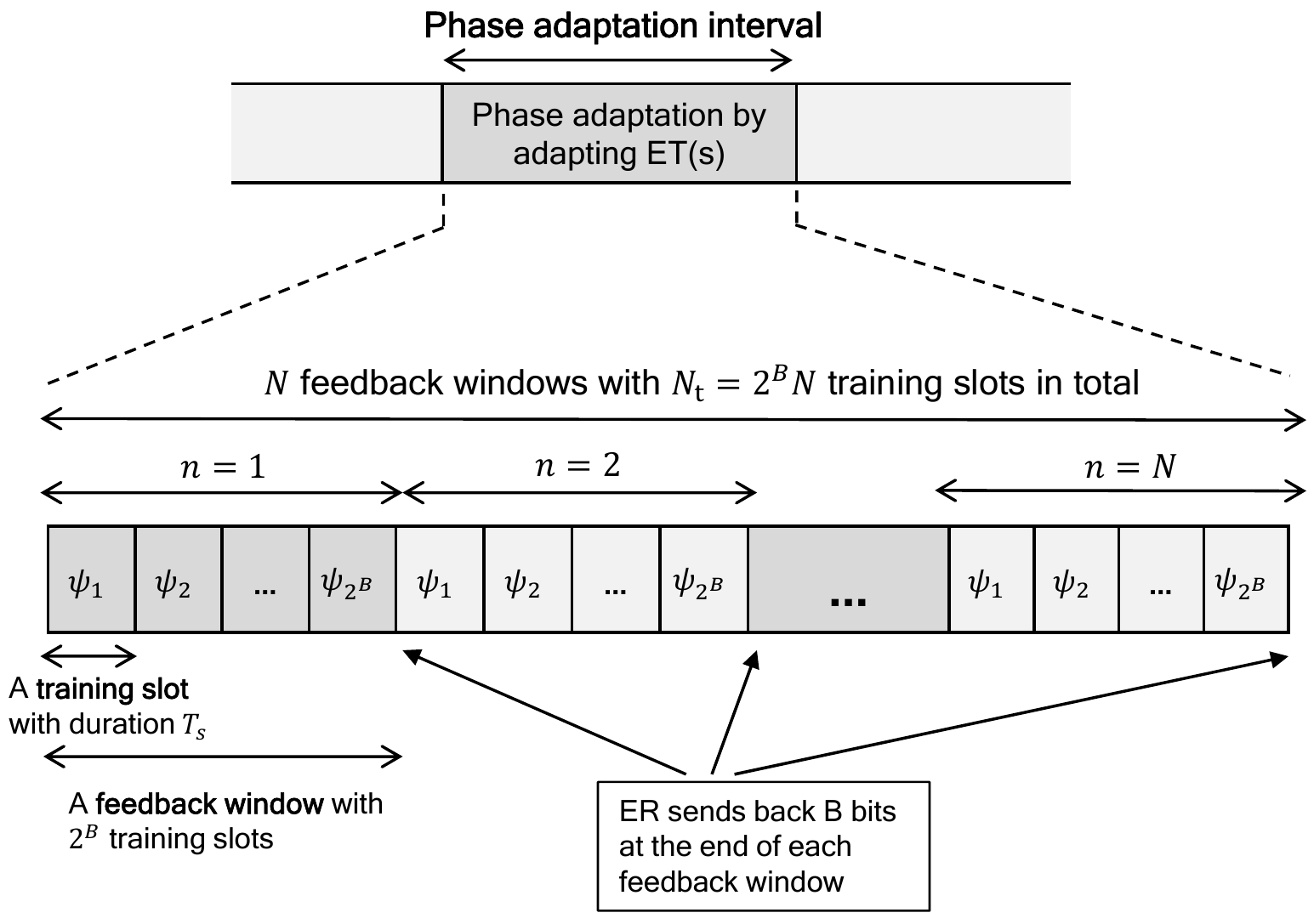}
\caption{Phase adaptation algorithm for each adapting ET.} \label{Fig:TimeSlots}
\end{figure}

We assume each phase adaptation interval consists of $N$ feedback windows of equal duration, as illustrated in Fig.~\ref{Fig:TimeSlots}, where the ER feeds back $B$ bits at the end of each feedback window, with $B\geq 1$. During each feedback window, each adapting ET$_m$, $m\in\cM_\text{A}$, transmits to ER with a sequence of $2^B$ phases, denoted by $\psi_1, \psi_2,...,\psi_{2^B}$. For convenience, we define the set of phase indices as $\cB = \{1,2,...,2^B\}$. It is further assumed that each transmit phase $\psi_b$, $b\in\cB$, is transmitted in a training slot with duration $T_s$, as shown in Fig.~\ref{Fig:TimeSlots}. The ER measures the average power received from the sum-signal of both adapting ETs and non-adapting ETs in each of the $2^B$ phases, denoted by $Q(\psi_b)$, $b\in \cB$, defined in \eqref{Eq:P_r 2}.\footnote{In practice, $T_s$ needs to be sufficiently large so that the ER can accurately measure the harvested power corresponding to each transmit phase. For simplicity, in this paper we assume that the power measurement at ER is perfect in each training slot.} Thus, a feedback window consists of $2^B$ training slots in total and its duration is $2^BT_s$. 

Let $b^*$ denote the index of the transmit phase that results in the maximum received power over $b\in\cB$, i.e., $b^{*} = \arg\max_{b\in\cB} Q(\psi_b)$, which according to \eqref{Eq:P_r 2} is equivalent to
\begin{equation} \label{Eq:Optimal b}
b^* = \arg\max_{b\in\cB} \cos(\varphi_\text{A}(\psi_b) + \varphi_\text{NA}).
\end{equation}
Using \eqref{Eq:Optimal b} and the fact that $\varphi_\text{A}(\psi^*) = -\varphi_\text{NA}$, we can obtain the following lemma. 
\begin{lemma} \label{Statement:ArgMax}
For $b^*$ defined in \eqref{Eq:Optimal b}, it also holds that $b^* = \arg\max_{b\in\cB} \cos(\psi^* - \psi_b)$. 
\end{lemma}
Lemma~\ref{Statement:ArgMax} implies that the transmit phase $\psi_{b^*}$ that results in the largest received power over $\cB$ should be closest to the target optimal phase $\psi^*$. At the end of each feedback window, the ER feeds back $B$ bits indicating $b^*$ (among $2^B$ phase indices) to each adapting ET$_m$, $m\in\cM_\text{A}$, as shown in Fig.~\ref{Fig:TimeSlots}. Let $\cA^{(n)}\subseteq [-\pi,\pi)$ denote the working set of each adapting ET$_m$, $m\in\cM_\text{A}$, for finding $\psi^*$ at the beginning of the $n$th feedback window, where $\psi^*\in \cA^{(n)}$. It is assumed that $\cA^{(1)} = [-\pi, \pi)$ and the transmit phases for the first $(n=1)$ feedback window are equally separated as given by $\psi_b = -\pi + \frac{\pi}{2^{B-1}}\l(b - \frac{1}{2}\r)$, $b\in \cB$, as illustrated in Fig.~\ref{Fig:Without_n1} for the case of $B=2$. With $b^{*}$ obtained from the ER feedback at the end of the $n$th feedback window, each adapting ET$_m$, $m\in\cM_\text{A}$, can infer that the target phase $\psi^*$ is not located in $\bigcup_{b\in\cB\backslash\{b^{*}\}}\{\theta\in [-\pi,\pi): \cos(\theta - \psi_{b^{*}}) < \cos(\theta - \psi_b)\}$, according to Lemma~\ref{Statement:ArgMax}. Consequently, at the end of the $n$th feedback window, each adapting ET$_m$, $m\in\cM_\text{A}$, updates its working set for the $(n+1)$th feedback window as 
\begin{equation} \label{Eq:WorkingSetWithout}
\cA^{(n+1)} =  \cA^{(n)}\backslash \bigcup_{b\in\cB\backslash\{b^{*}\}}\{\theta\in [-\pi,\pi): \cos(\theta - \psi_{b^{*}}) < \cos(\theta - \psi_b)\}.  
\end{equation}
Define $x = \inf_{\theta\in\cA^{(n+1)}}\theta$ and $y = \sup_{\theta\in\cA^{(n+1)}}\theta$, with $-\pi\leq x \leq y < \pi$. Finally, each adapting ET$_m$, $m\in\cM_\text{A}$, updates the transmit phases for the next window as
\begin{equation} \label{Eq:PhaseWithout}
\psi_b = x + \frac{y - x}{2^{B+1}} + \frac{y- x}{2^B}(b-1), \quad b\in\cB,
\end{equation}
by equally dividing the set $\cA^{(n+1)}$ into $2^B$ subsets, as illustrated in Figs.~\ref{Fig:Without_n2} and \ref{Fig:Without_n3} for the cases of $\cA^{(2)}$ and $\cA^{(3)}$, respectively, with $B=2$. The above procedure is repeated for $N$ feedback windows, and at the end of the $N$th window each adapting ET$_m$, $m\in\cM_\text{A}$, sets $\bar{\phi}_m = \psi_{b^*}$. The above phase adaptation algorithm is summarized in \textbf{Algorithm 1} (A1). For convenience, we also refer to this algorithm as ``$B$-bit feedback without memory".

It is worth noting that the set of transmit phases $\psi_1,...,\psi_{2^B}$ at the end of the $n$th feedback window equally divides the set $\cA^{(n)}$ into $2^B$ subsets, and the subset corresponding to $\psi_{b^*}$ will become $\cA^{(n+1)}$ in the subsequent window; as a result, the size of the working set $\cA^{(n+1)}$ becomes $1/2^B$ of that of $\cA^{(n)}$ (see Fig.~\ref{Fig:Algorithm_Without}), which exponentially decreases over $N$. The convergence performance of (A1) will be analyzed in more detail in Section~\ref{Section:Performance}.

\begin{figure}
\centering
\subfigure[$n=1$]{
\centering
\includegraphics[width=6cm]{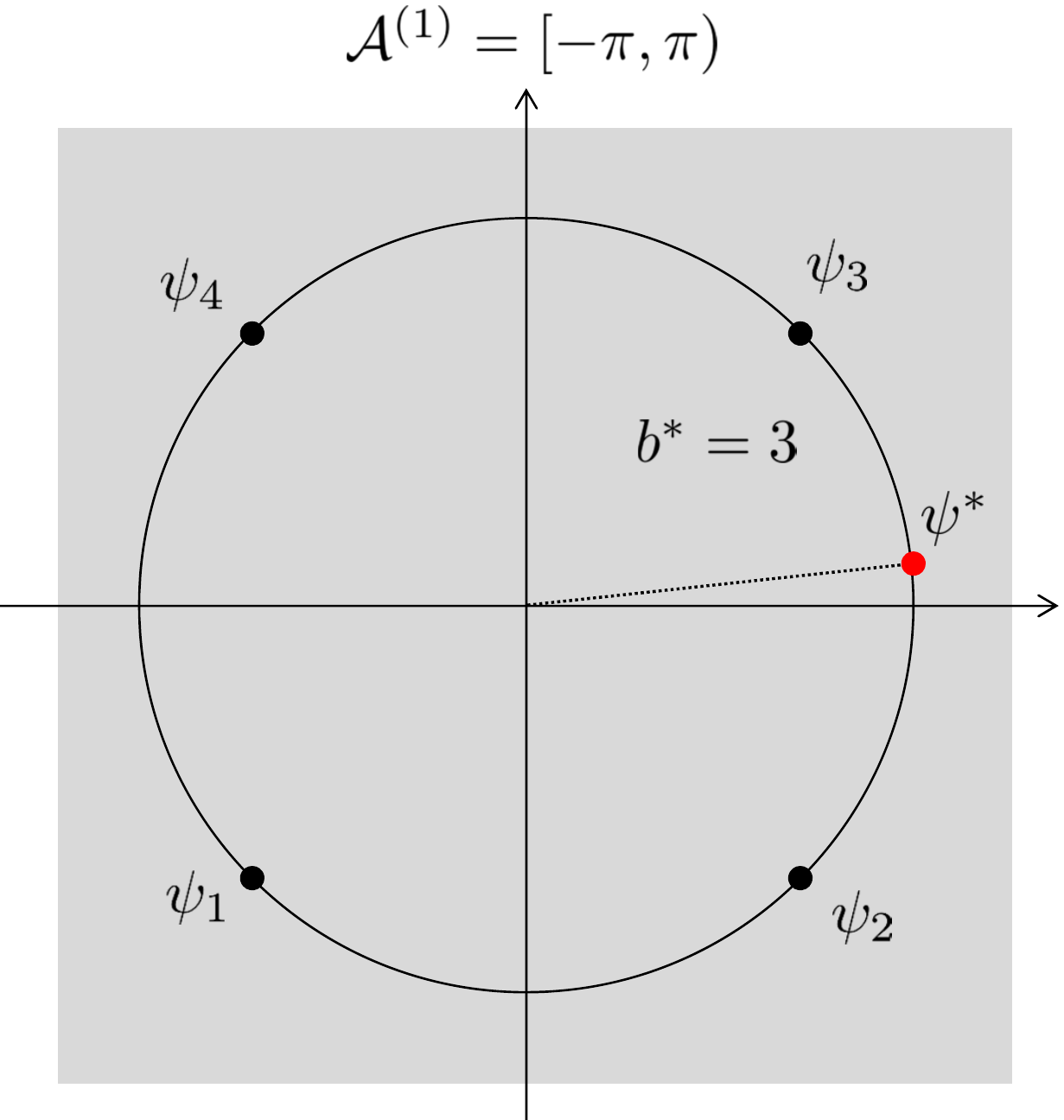}\label{Fig:Without_n1}} 
\subfigure[$n=2$]{
\centering
\includegraphics[width=6cm]{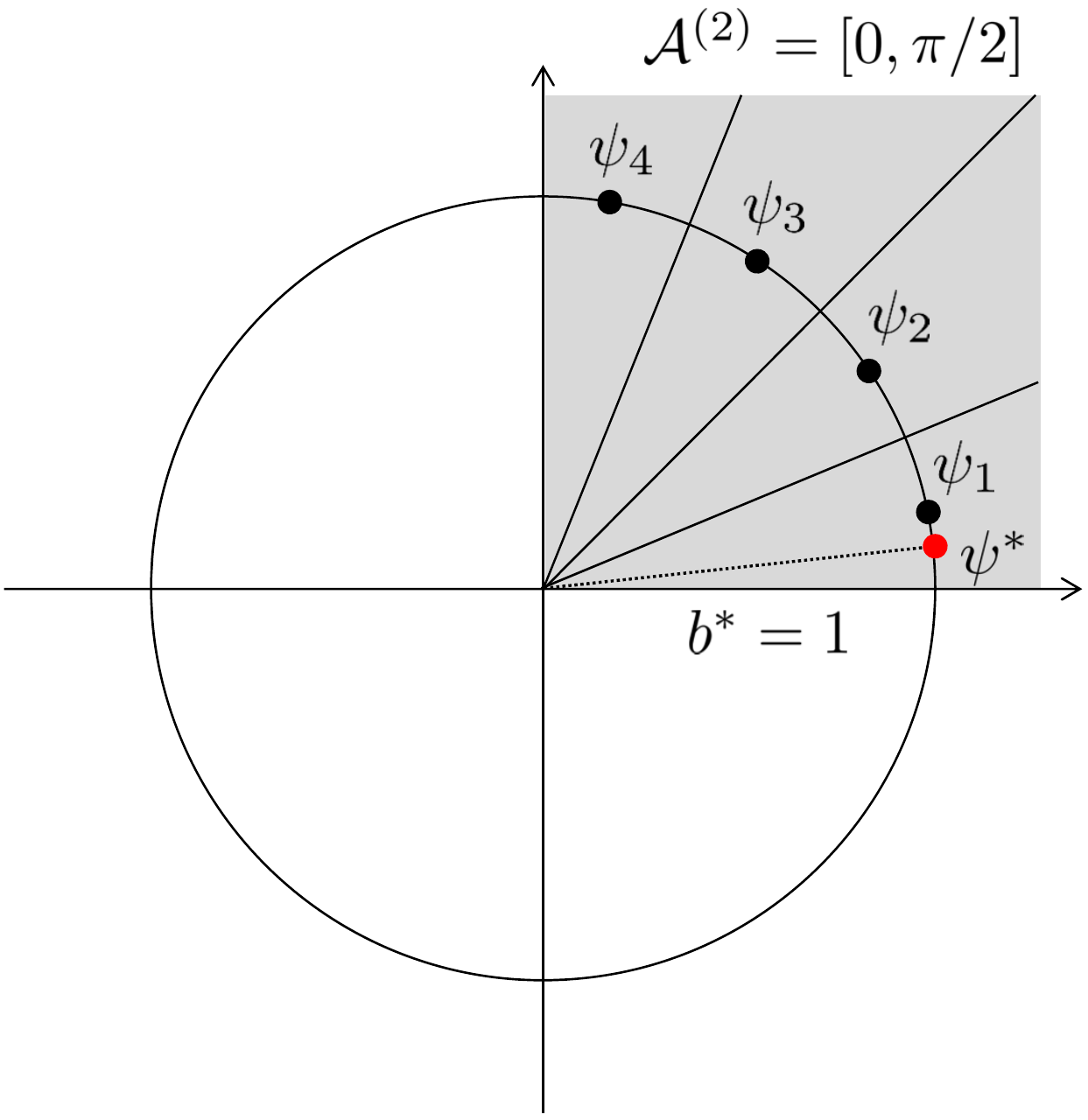}\label{Fig:Without_n2}} 
\subfigure[$n=3$]{
\centering
\includegraphics[width=6cm]{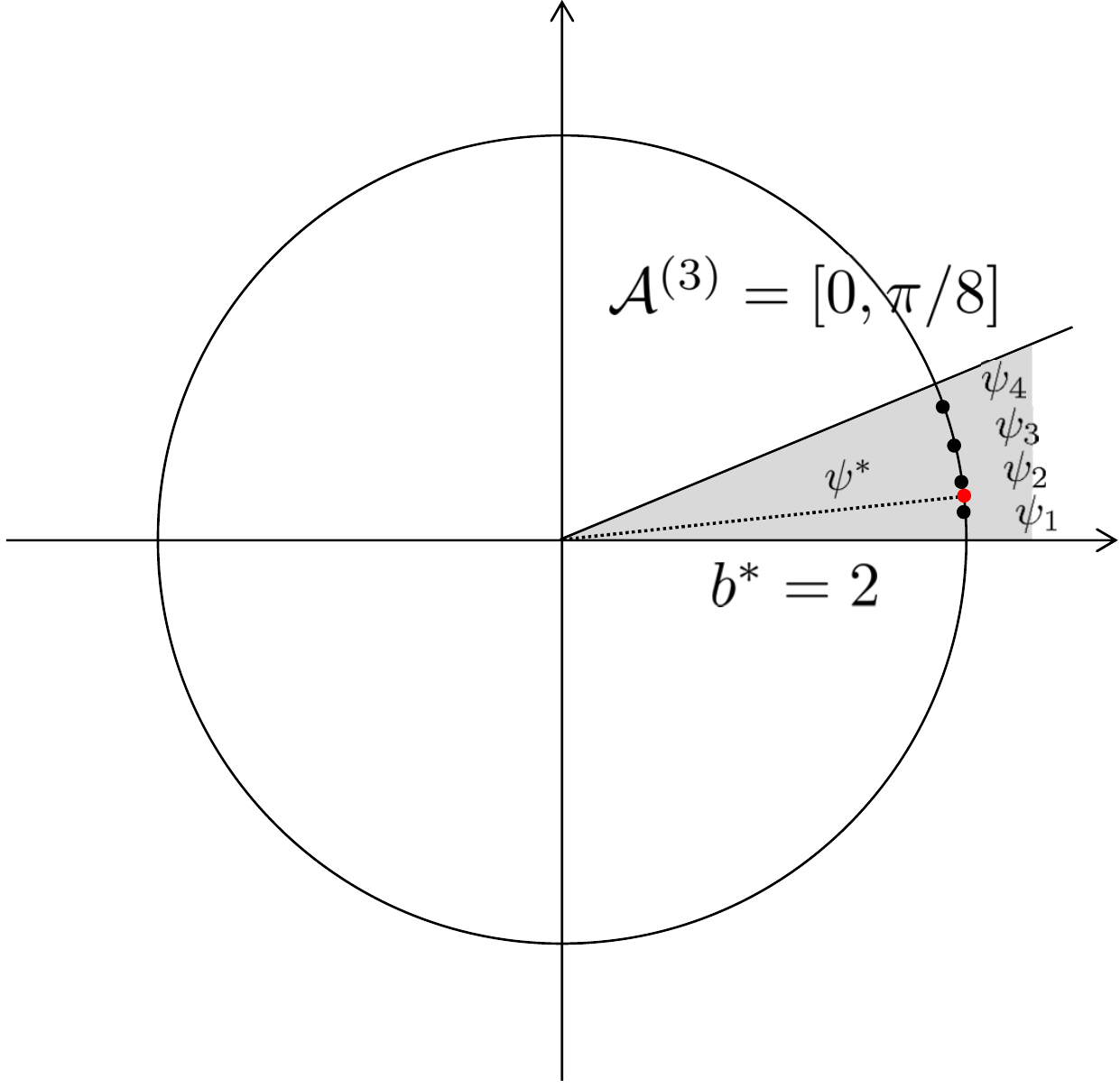}\label{Fig:Without_n3}} 
\caption{Illustration of (A1) for locating the target phase $\psi^*$, for $n=1,2,3$ with $B=2$. The shaded regions indicate the working set $\cA^{(n)}$ at the beginning of the $n$th feedback window.}
\label{Fig:Algorithm_Without}
\end{figure}

\begin{algorithm}
\caption{(A1): $B$-bit feedback without memory}
\begin{algorithmic}[1]
\State \textbf{Initialize:} Each ET$_m$, $m\in\cM_\text{A}$, sets $\cA^{(1)} = [-\pi, \pi)$, and $\psi_b = -\pi + \frac{2\pi}{2^{B+1}} + \frac{2\pi}{2^B}(b-1)$, $b\in\cB$.
\For{ $n=1 : N$}
\State Each ET$_m$, $m\in\cM_\text{A}$, sequentially transmits with $\phi_m = \psi_b$,  $b\in\cB$.
\State ER computes $b^* = \arg\max_{b\in\cB} Q_m(\psi_b)$, and feeds back its $B$-bit index to each ET$_m$, $m\in\cM_\text{A}$.
\State Each ET$_m$, $m\in\cM_\text{A}$, sets $\cA^{(n+1)} = \cA^{(n)}\backslash \bigcup_{b\in\cB\backslash\{b^*\}}\{\theta\in [-\pi,\pi): \cos(\theta - \psi_{b^*}) < \cos(\theta - \psi_b)\}$.
\State Each ET$_m$, $m\in\cM_\text{A}$, sets $x = \inf_{\theta\in\cA^{(n+1)}}\theta$ and $y = \sup_{\theta\in\cA^{(n+1)}}\theta$ 
\State Each ET$_m$, $m\in\cM_\text{A}$, sets $\psi_b$'s as in \eqref{Eq:PhaseWithout}, $b\in\cB$.
\EndFor
\State Each ET$_m$, $m\in\cM_\text{A}$, sets $\bar{\phi}_m = \psi_{b^*}$.
\end{algorithmic}
\end{algorithm}

\subsection{Phase Adaptation Algorithm With Memory} \label{Subsection:A2}
In this subsection, we extend the phase adaptation algorithm in Section~\ref{Subsection:A1} by considering the case that the ER has memory, i.e., it stores and updates the maximum harvested power up to each feedback window.
 
\begin{figure} 
\centering
\includegraphics[width=6cm]{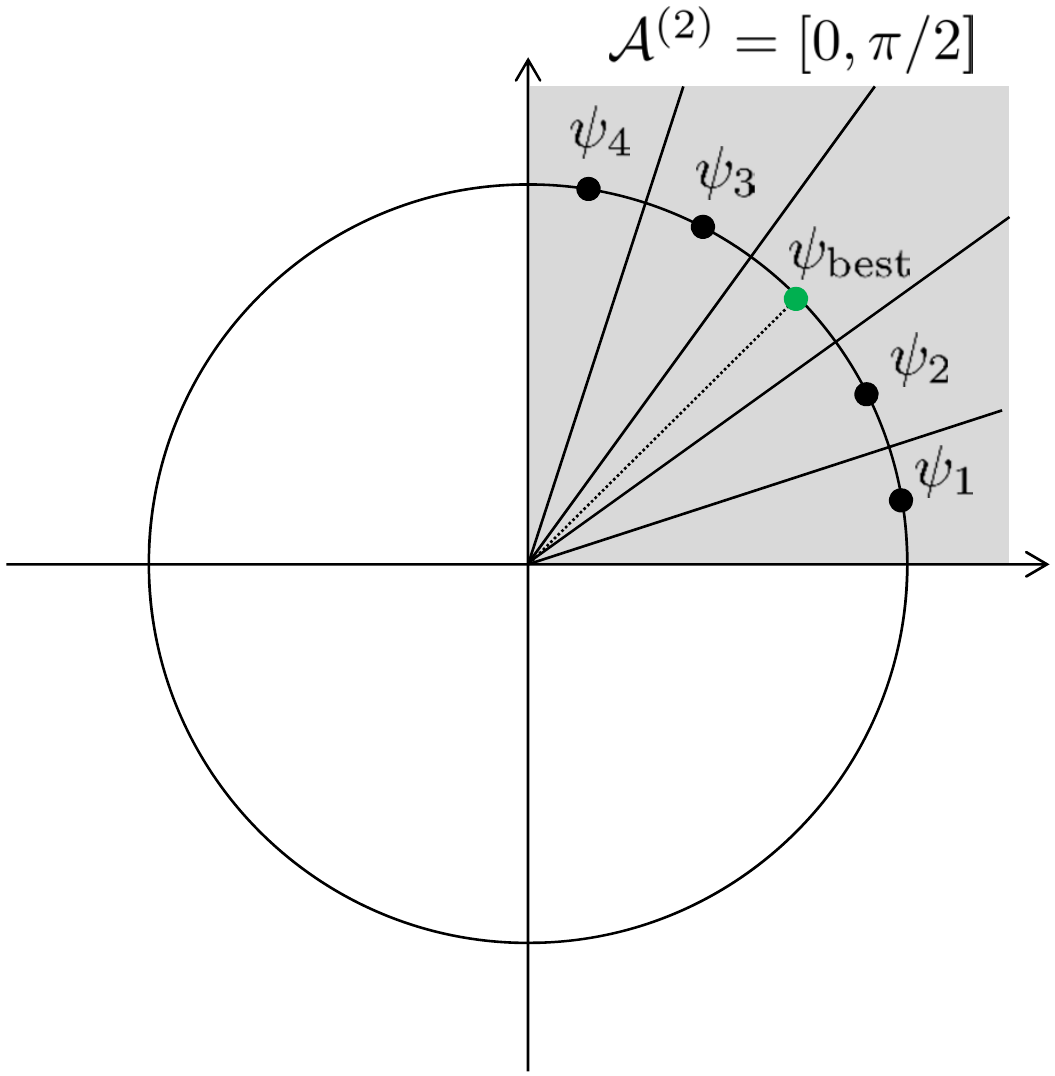}
\caption{Illustration of the proposed phase adaptation algorithm with memory (A2) for $n=2$ with $B=2$, while the working set and the transmit phases for $n=1$ are the same as the case without memory shown in Fig.~\ref{Fig:Without_n1} .} \label{Fig:With_n12}
\end{figure}

Specifically, at the first feedback window ($n=1$), each adapting ET$_m$, $m\in\cM_\text{A}$,   transmits to the ER with a sequence of $2^B$ equally-spaced phases $\psi_1, \psi_2,...,\psi_{2^B}$, each with duration $T_s$, similar to (A1) as shown in Fig.~\ref{Fig:Without_n1} with $B=2$. Then, at the end of this window, the ER feeds back $B$ bits indicating $b^* = \arg\max_{b\in\cB} Q(\psi_b)$ to each adapting ET$_m$, $m\in\cM_\text{A}$, or equivalently $b^* = \arg\max_{b\in\cB} \cos(\psi^* - \psi_b)$ according to Lemma~\ref{Statement:ArgMax}. Unlike (A1), the value of $Q(\psi_{b^*})$ is stored in the memory of the ER as the maximum harvested power so far, denoted by $Q_\text{best}$, i.e., $Q_\text{best} = Q(\psi_{b^*})$. After receiving $b^*$, each adapting ET$_m$ updates its working set to be $\cA^{(2)} = \cA^{(1)}\backslash \bigcup_{b\in\cB\backslash\{b^*\}}\{\theta\in [-\pi,\pi): \cos(\theta - \psi_{b^*}) < \cos(\theta - \psi_b)\}$, same as (A1); while it also denotes the best transmit phase so far as $\psi_\text{best} = \psi_{b^*}$. Let $x = \inf_{\theta\in\cA^{(n+1)}}\theta$ and $y = \sup_{\theta\in\cA^{(n+1)}}\theta$. By exploiting the fact that the ER knows the value of $Q(\psi_{b^*})$ which is recorded as $Q_\text{best}$, at the end of the first feedback window, each adapting ET$_m$, $m\in\cM_\text{A}$, updates the transmit phases for the next window ($n=2$) as
\begin{equation} \label{Eq:PhaseCase1}
\psi_b = \l\{\begin{aligned}
&x + \frac{y - x}{2(2^{B}+1)} + \frac{y - x}{2^B+1}(b-1), \quad \text{if } b=1,...,2^{B-1}, \\
&x + \frac{y - x}{2(2^{B}+1)} + \frac{y - x}{2^B+1}b, \;\quad\qquad \text{if } b=2^{B-1}+1,...,2^B,
\end{aligned}
\r.
\end{equation}
which is illustrated in Fig.~\ref{Fig:With_n12} for the case of $B=2$. Note that the set of phases $\psi_1,...,\psi_{2^B}$ in \eqref{Eq:PhaseCase1} together with $\psi_\text{best}$ equally divide $\cA^{(2)}$ into $2^B +1$ subsets. This is more efficient compared to the previous case without memory where only $2^B$ subsets are obtained for $n=2$ (see Fig.~\ref{Fig:Without_n2}).

\begin{figure}
\centering
\subfigure[]{
\centering
\includegraphics[width=5cm]{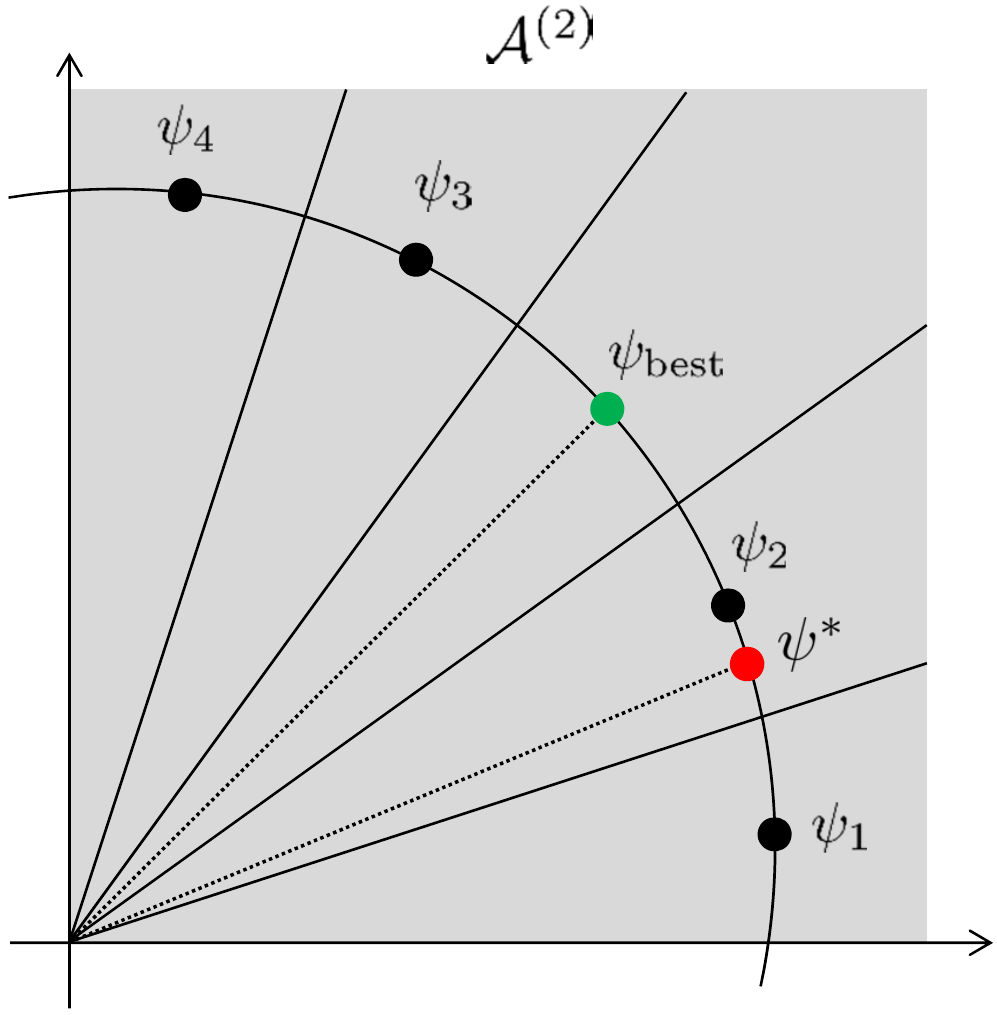}\label{Fig:Case1_1}} 
\subfigure[]{
\centering
\includegraphics[width=5cm]{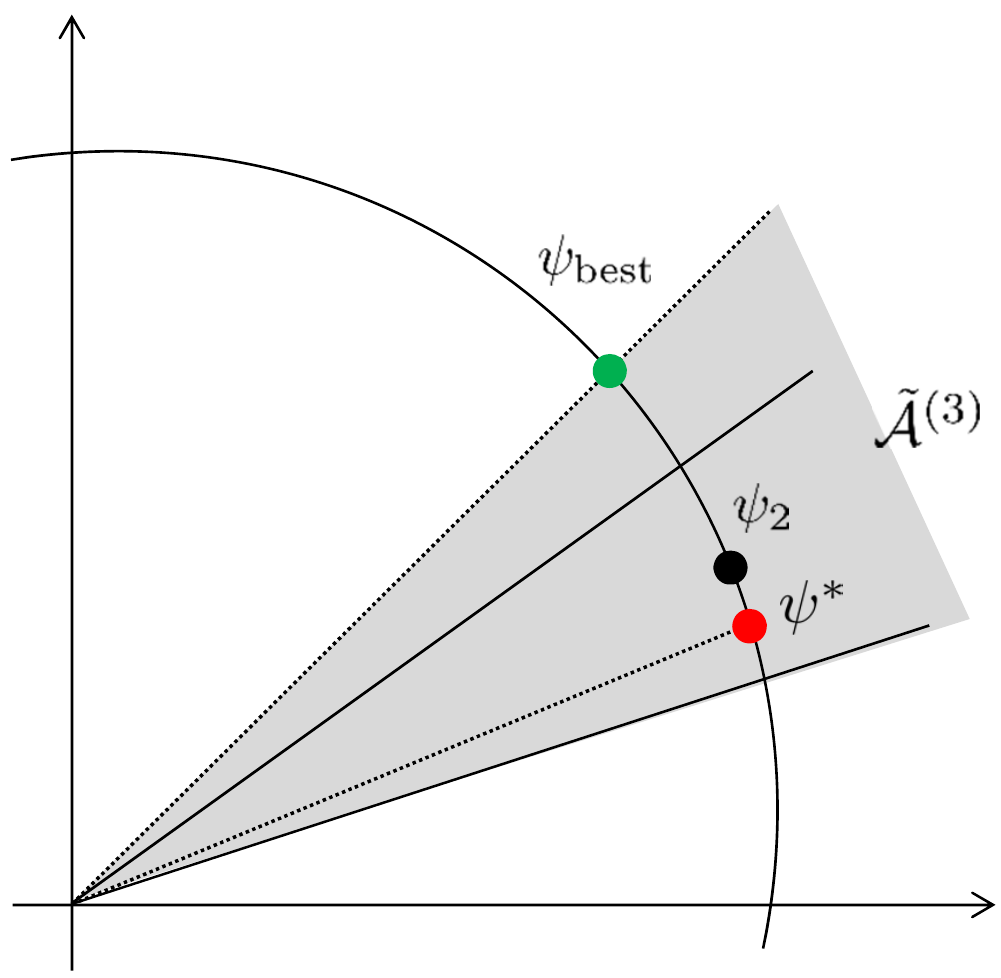}\label{Fig:Case1_2}} 
\subfigure[]{
\centering
\includegraphics[width=10cm]{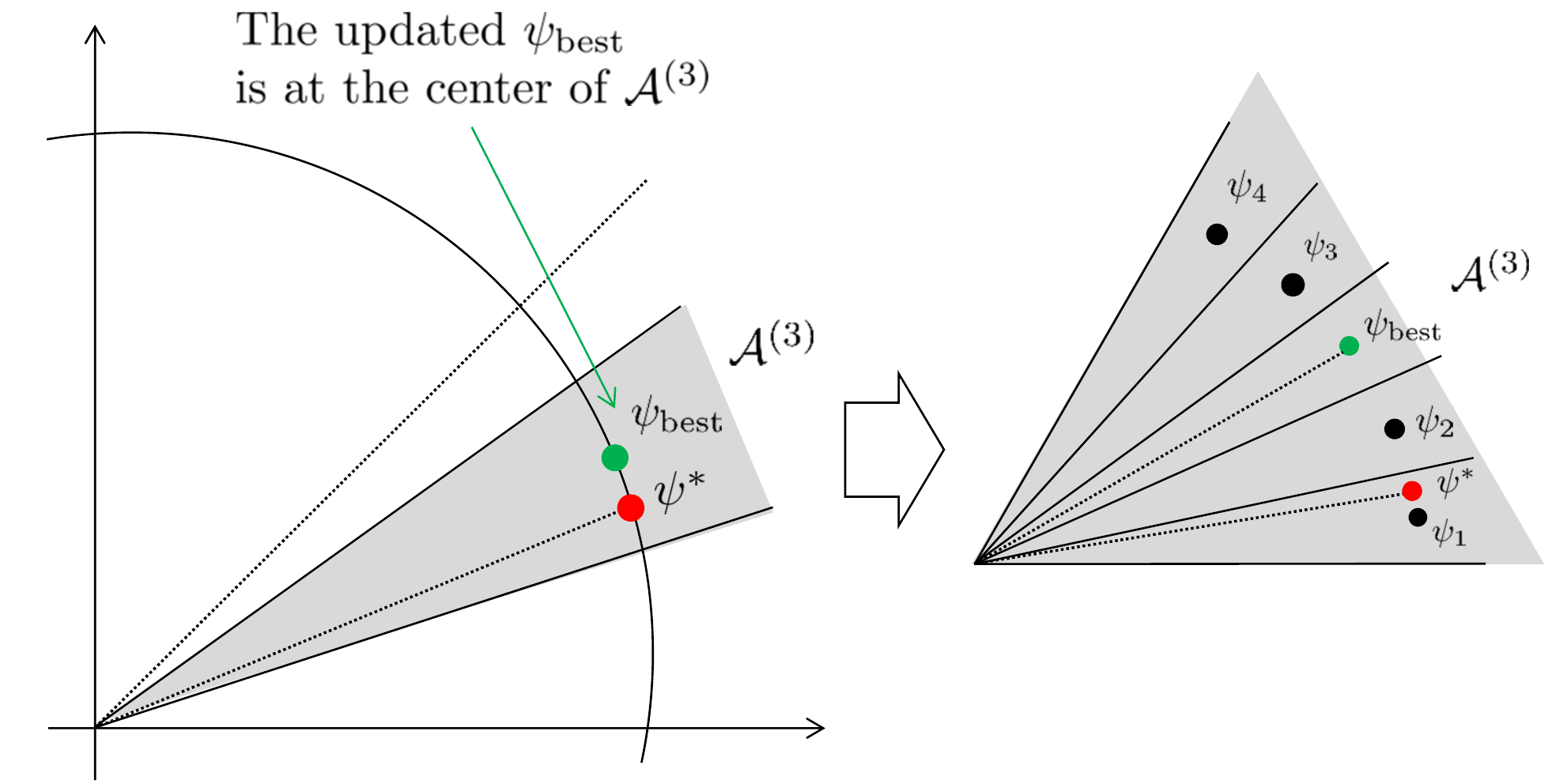}\label{Fig:Case1_3}} 
\caption{Illustration of Case 1 in (A2) for $n=2$. (a) $\cA^{(2)}$; (b) $\tilde{\cA}^{(3)}$; (c) Case 1: $\cA^{(3)}$ with $\psi_\text{best} = \psi_2$, $Q_\text{best} = Q(\psi_2)$, and updated $\psi_b$'s for feedback window $n=3$.}
\label{Fig:Case1}
\end{figure}

\begin{figure}
\centering
\subfigure[]{
\centering
\includegraphics[width=5cm]{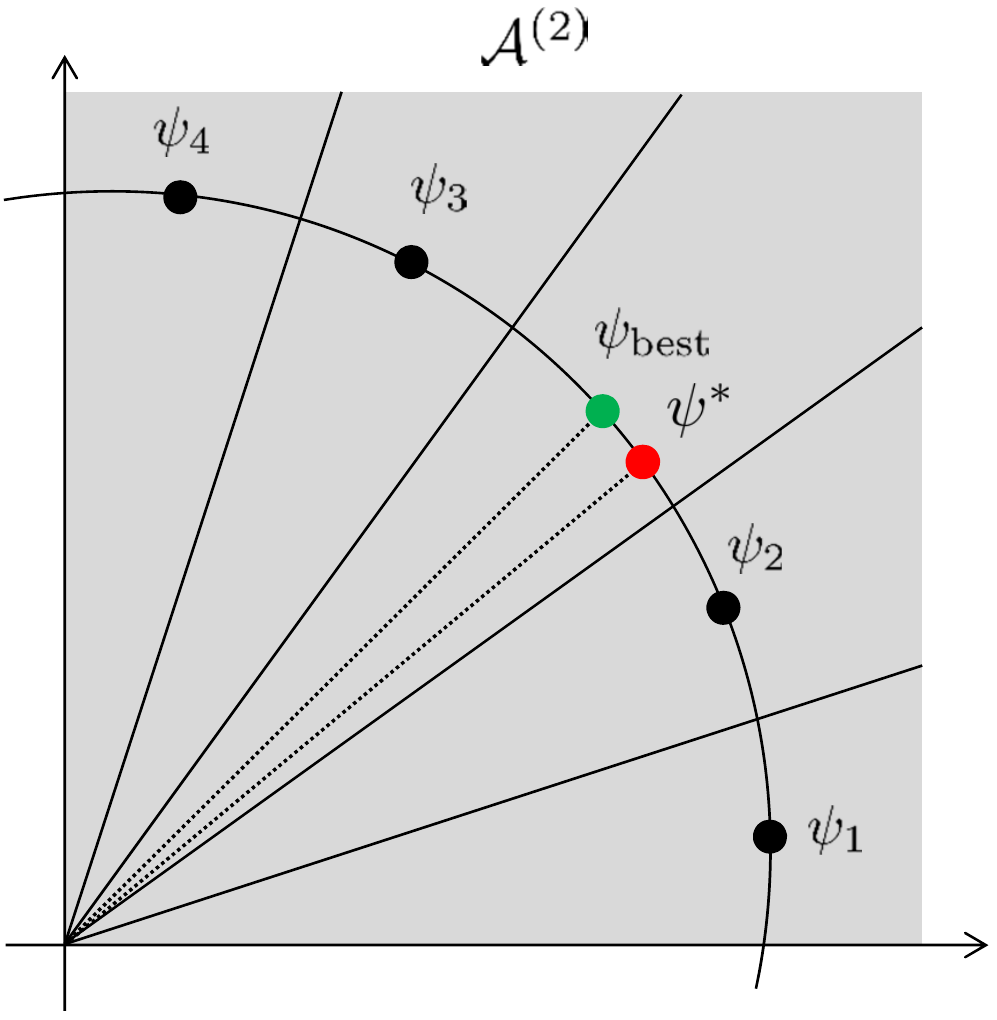}\label{Fig:Case2_1}} 
\subfigure[]{
\centering
\includegraphics[width=5cm]{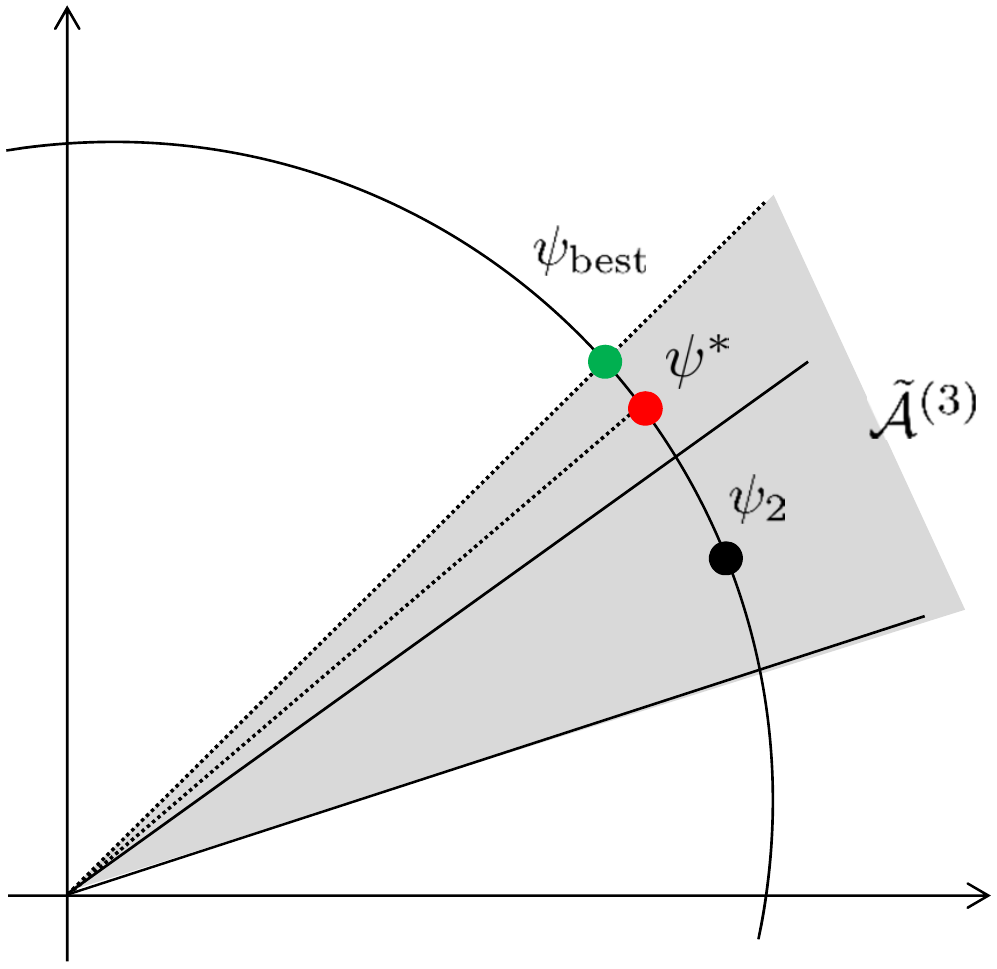}\label{Fig:Case2_2}} 
\subfigure[]{
\centering
\includegraphics[width=10cm]{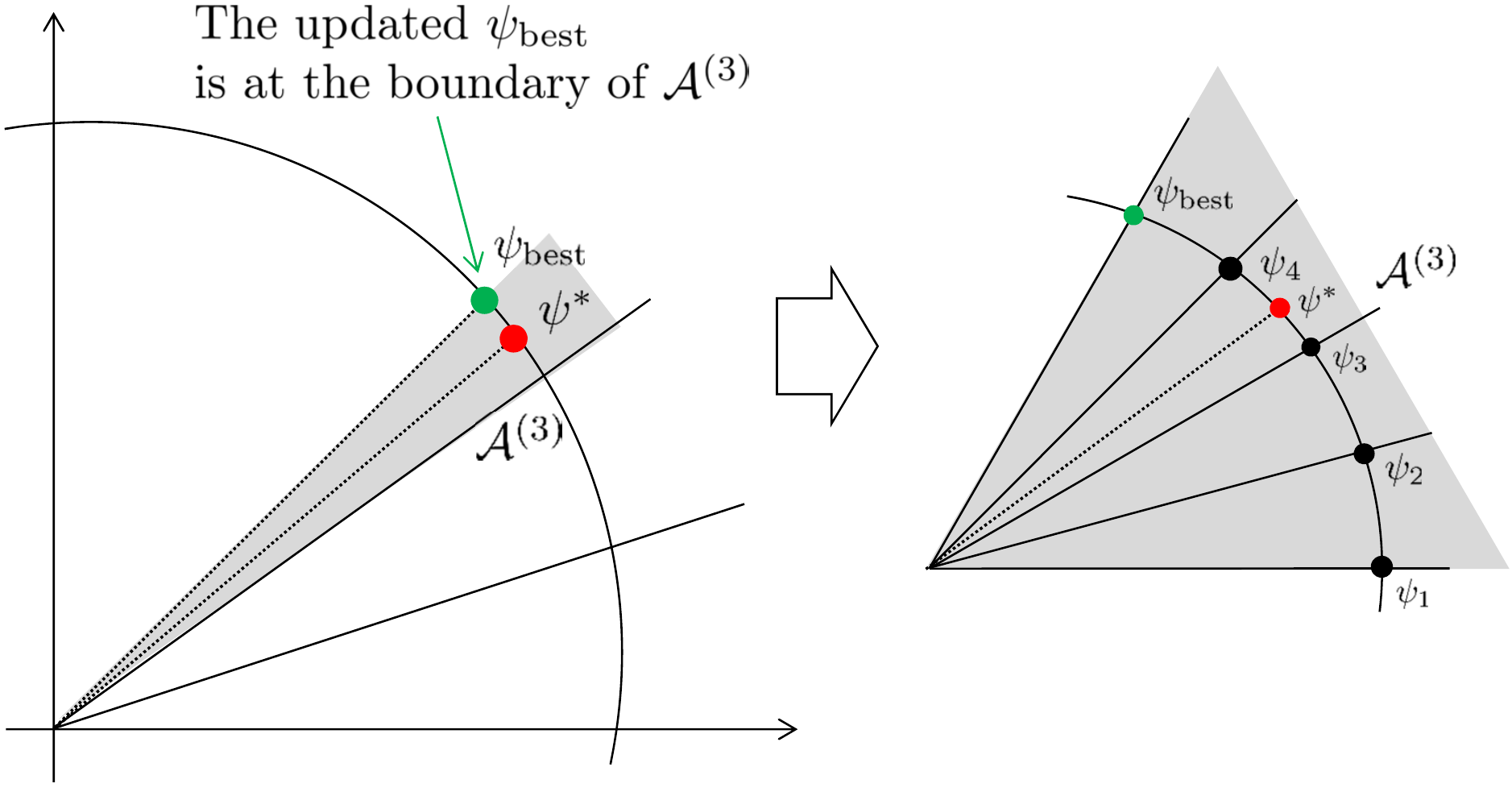}\label{Fig:Case2_3}} 
\caption{Illustration of Case 2 in (A2) for $n=2$. (a) $\cA^{(2)}$; (b) $\tilde{\cA}^{(3)}$; (c) Case 2: $\cA^{(3)}$ with $\psi_\text{best}$ and $Q_\text{best}$ unchanged, and updated $\psi_b$'s for feedback window $n=3$.}
\label{Fig:Case2}
\end{figure}

From the second feedback window ($n=2$), each adapting ET$_m$, $m\in\cM_\text{A}$, transmits with the sequence of $2^B$ phases given in \eqref{Eq:PhaseCase1}, and then the ER feeds back $B$ bits indicating $b^* = \arg\max_{b\in\cB} \cos(\psi^* - \psi_b)$ to each adapting ET$_m$, $m\in\cM_\text{A}$, same as (A1) for the case without memory. However, in addition to the $B$-bit feedback, from the window $n=2$, the ER compares the maximum harvested power in its memory, i.e., $Q_\text{best}$, with the largest harvested power in the current window, i.e., $Q(\psi_{b^*})$, and feeds back one more bit to each adapting ET$_m$ to help further reduce its working set. Specifically, by using the first $B$-bit feedback indicating $b^*$, each adapting ET$_m$  sets
\begin{equation} \label{Eq:B-bits}
\tilde{\cA}^{(n+1)} = \cA^{(n)}\backslash \bigcup_{b\in\cB\backslash\{b^*\}}\{\theta\in [-\pi,\pi): \cos(\theta - \psi_{b^*}) < \cos(\theta - \psi_b)\}.
\end{equation}
Next, by utilizing the additional one-bit feedback indicating either $Q(\psi_{b^*}) \geq Q_\text{best}$ or $Q_m(\psi_{b^*}) < Q_\text{best}$, each adapting ET$_m$ further reduces its working set as follows. If $Q(\psi_{b^*}) \geq Q_\text{best}$, referred to as Case 1, ET$_m$ sets
\begin{equation} \label{Eq:Another One-bit1}
\cA^{(n+1)} = \tilde{\cA}^{(n+1)}\backslash \{\theta\in [-\pi,\pi): \cos(\theta - \psi_{b^*}) < \cos(\theta - \psi_\text{best})\},
\end{equation}
and ET$_m$ and the ER update $\psi_\text{best} = \psi_{b^*}$ and $Q_\text{best} = Q(\psi_{b^*})$, respectively, since $\psi_{b^*}$ now becomes the best phase. On the other hand, if $Q(\psi_{b^*}) < Q_\text{best}$, referred to as Case 2, ET$_m$ sets
\begin{equation} \label{Eq:Another One-bit2}
\cA^{(n+1)} = \tilde{\cA}^{(n+1)}\backslash \{\theta\in [-\pi,\pi): \cos(\theta - \psi_{b^*}) > \cos(\theta - \psi_\text{best})\},
\end{equation}
while both $\psi_\text{best}$ and $Q_\text{best}$ remain unchanged. The above two cases are shown in Figs.~\ref{Fig:Case1_3} and ~\ref{Fig:Case2_3}, respectively, for $\cA^{(3)}$. Further, the transmit phases for the next window are determined for the two cases as follows.

\begin{itemize}
\item Case 1: In this case, the updated $\psi_\text{best}$ is located at the center of $\cA^{(n+1)}$, i.e., $\psi_\text{best} = (x + y)/2$, as illustrated in Fig.~\ref{Fig:Case1_3}. As a result, each adapting ET$_m$, $m\in\cM_\text{A}$, sets $\psi_b$'s according to \eqref{Eq:PhaseCase1}, as shown in Fig.~\ref{Fig:Case1_3}.
 
\item Case 2: In this case, the updated $\psi_\text{best}$ is located at one of the two boundary points of $\cA^{(n+1)}$, i.e., either $\psi_\text{best} = x$ or $\psi_\text{best} = y$, as illustrated in Fig.~\ref{Fig:Case2_3}. Thus, each adapting ET$_m$, $m\in\cM_\text{A}$, sets $\psi_b$, $b\in\cB$, as 
\begin{equation} \label{Eq:PhaseCase2}
\psi_b = \l\{\begin{aligned}
& x + \frac{y - x}{2^{B}}b,  &\quad \mbox{if }\psi_\text{best} = x, \\
& x + \frac{y - x}{2^{B}}(b-1), &\quad \mbox{if }\psi_\text{best} = y,
\end{aligned}
\r.
\quad 
\end{equation}
as illustrated in Fig.~\ref{Fig:Case2_3}.
\end{itemize}

It can be verified that in the subsequent feedback windows ($n\geq 3$), either Case 1 or Case 2 may occur. Thus, the above procedure is repeated for $N$ feedback windows, and finally each adapting ET$_m$, $m\in\cM_\text{A}$, sets $\bar{\phi}_m = \psi_\text{best}$. The above phase adaptation algorithm is summarized in \textbf{Algorithm 2} (A2). Since this algorithm utilizes the memory at the ER and one additional feedback bit comparing $Q_\text{best}$ with $Q(\psi_{b^*})$, it is also referred to as ``$(B+1)$-bit feedback with memory". 

\begin{algorithm}
\caption{(A2): $(B+1)$-bit feedback with memory}
\begin{algorithmic}[1]
\State \textbf{Initialize:} Each ET$_m$, $m\in\cM_\text{A}$, sets $\cA^{(1)} = [-\pi, \pi)$, $\psi_b = -\pi + \frac{2\pi}{2^{B+1}} + \frac{2\pi}{2^B}(b-1)$, $b\in\cB$, $\psi_\text{best} = \psi_1$. ER sets $Q_\text{best} = 0$.
\For{ $n=1 : N$}
\State Each ET$_m$, $m\in\cM_\text{A}$, sequentially transmits with $\phi_m = \psi_b$,  $b\in\cB$.
\State ER computes $b^* = \arg\max_{b\in\cB} Q_m(\psi_b)$, and feeds back its $B$-bit index to each ET$_m$, $m\in\cM_\text{A}$.

\State Each ET$_m$, $m\in\cM_\text{A}$, first sets $\tilde{\cA}^{(n+1)}$ as in \eqref{Eq:B-bits}.
\State ER also compares $Q_m(\psi_{b^*})$ and $Q_\text{best}$, and feeds back the corresponding one-bit to each ET$_m$, $m\in\cM_\text{A}$.
\If{ $Q_m(\psi_{b^*}) \geq Q_\text{best}$}
\State Each ET$_m$, $m\in\cM_\text{A}$, sets $\cA^{(n+1)}$ as in \eqref{Eq:Another One-bit1}, and $\psi_\text{best} = \psi_{b^*}$.
\State ER sets $Q_\text{best} = Q_m(\psi_{b^*})$.
\Else
\State Each ET$_m$, $m\in\cM_\text{A}$, sets $\cA^{(n+1)}$ as in \eqref{Eq:Another One-bit2}.
\EndIf
\State Each ET$_m$, $m\in\cM_\text{A}$, sets $x = \inf_{\theta\in\cA^{(n+1)}}\theta$ and $y = \sup_{\theta\in\cA^{(n+1)}}\theta$.
\If{ $\psi_\text{best} = \frac{x+y}{2}$ (Case 1)}
\State Each ET$_m$, $m\in\cM_\text{A}$, sets $\psi_b$'s as in \eqref{Eq:PhaseCase1}, $b\in\cB$.
\Else{ (Case 2)}
\State Each ET$_m$, $m\in\cM_\text{A}$, sets $\psi_b$'s as in \eqref{Eq:PhaseCase2}, $b\in\cB$.
\EndIf
\EndFor
\State Each ET$_m$, $m\in\cM_\text{A}$, sets $\bar{\phi}_m = \psi_\text{best}$.
\end{algorithmic}
\end{algorithm}

\section{Performance Analysis} \label{Section:Performance}
In this section, we first analyze and compare the convergence performances of the two phase adaptation algorithms (A1) without ER memory and (A2) with ER memory presented in Section~\ref{Section:PhaseAdaptation}, based on which the efficiency of the distributed EB protocol with sequential training is then characterized.

\subsection{Convergence Performance of (A1) and (A2)} \label{Subsection:Performance_Algorithm}
The objective of this subsection is to investigate how accurately the proposed phase adaptation algorithms can estimate the target phase $\psi^*$ after each phase adaptation interval, given $N_\text{t}$ training slots per phase adaptation interval (see Fig.~\ref{Fig:TimeSlots}). To analyze the performance for (A1), i.e., the algorithm of $B$-bit feedback without memory, as well as (A2), i.e., the algorithm of $(B+1)$-bit feedback with memory, we define the phase-error between the final estimated phase by (A1) or (A2) at each adapting ET$_m$, $m\in\cM_\text{A}$, i.e., $\bar{\phi}_m$, and the optimal phase $\psi^*$ as
\begin{equation} \label{Eq:Error}
e = \bar{\phi}_m - \psi^* .
\end{equation}

\subsubsection{Phase-Error Upper Bound for (A1)}
In (A1), as mentioned in Section~\ref{Subsection:A1}, the size of the working set $\cA^{(n)}$ is reduced by $1/2^B$ per feedback window (see Fig.~\ref{Fig:Algorithm_Without}). Moreover, since each adapting ET$_m$, $m\in \cM_\text{A}$, finally sets $\bar{\phi}_m  =\psi_{b^*}$, given $N_\text{t}$ total training slots per phase adaptation interval and $\cA^{(1)} = [-\pi,\pi)$, it can be verified that the absolute value of the phase-error at each ET$_m$ is upper-bounded by
\begin{equation} \label{Eq:ErrorBound}
|e| \leq \pi \l(\frac{1}{2^B}\r)^{N_\text{t}/2^B},
\end{equation}
where we assume that $N = N_\text{t} / 2^B$ is an integer for simplicity. In other words, after $N_\text{t}/2^B$ feedback windows for phase adaptation, the maximum error between the estimated phase and the target phase is upper-bounded by the right-hand side (RHS) of \eqref{Eq:ErrorBound}, which exponentially decreases to zero with increasing $N_\text{t}$. Thus, the estimated phase $\bar{\phi}_m$ of (A1) converges to $\psi^*$ exponentially fast with $N_\text{t}$ given fixed $B$.

Next, we investigate the effect of different values of $B$ on the convergence speed of (A1). It can be easily shown that $B=1$ and $B=2$ yield the same value of the RHS of \eqref{Eq:ErrorBound}, which is further an increasing function over $B\geq 2$ with fixed $N_\text{t}$ (this can be easily verified by taking the logarithm and then the derivative with respect to $B$ in \eqref{Eq:ErrorBound}). Thus, we obtain the following proposition.
\begin{proposition} \label{Statement:ErrorBound_A1}
Given $N_\text{t}\geq 2^B$ total training slots, $B=1$ or $B=2$ yields the minimum phase-error upper bound in (A1). 
\end{proposition}

Proposition~\ref{Statement:ErrorBound_A1} suggests that for (A1), using only one-bit or two-bit feedback per feedback window is most efficient in terms of phase estimation accuracy. Without loss of generality, we assume $B=1$ as the most efficient design for (A1) in the sequel.

\subsubsection{Phase-Error Upper Bound for (A2)}
Unlike (A1), two different cases may occur in (A2) depending on the location of $\psi^*$ for feedback windows $n\geq 2$, as described in Section~\ref{Subsection:A2}. As a result, the phase-error in \eqref{Eq:Error} by (A2) can take different values depending on different combinations of Cases 1 and 2 over the $N$ feedback windows. In the following lemma, we analyze the phase-error upper bound for (A2).

\begin{lemma} \label{Statement:ErrorBound_A2_1}
Given $N_\text{t}\geq 2^B$, the phase-error upper bound for (A2) is given by
\begin{equation} \label{Eq:A2_Worst}
|e| \leq  \frac{\pi}{2^B} \l(\frac{1}{2^B +1}\r)^{\frac{N_\text{t}}{2^B}-1}.
\end{equation}
\end{lemma}
\begin{proof}
See Appendix~\ref{Proof:ErrorBound_A2_1}.
\end{proof}

It is worth noting that the phase-error upper bound given in \eqref{Eq:A2_Worst} for (A2) exponentially decreases with $N_\text{t}$ given fixed $B$, similar to (A1).

The following proposition shows that $B=1$ achieves the best efficiency for (A2).

\begin{proposition}  \label{Statement:ErrorBound_A2_2}
Given $N_\text{t}\geq 2^B$, the phase-error upper bound given in \eqref{Eq:A2_Worst} is an increasing function over $B\geq 1$. 
\end{proposition}
\begin{proof}
See Appendix~\ref{Proof:ErrorBound_A2_2}.
\end{proof}

It can be inferred from Proposition~\ref{Statement:ErrorBound_A2_2} that $B=1$ is  most efficient in locating the target phase $\psi^*$ in (A2) with fixed $N_\text{t}$. Since in (A2), one additional feedback bit is needed as compared to (A1), we conclude that two bits per feedback is optimal for (A2), or $B=1$ is optimal for both (A1) and (A2). This is quite a surprising result as using large number of bits per ER feedback does not help improve the performance of our proposed phase adaptation algorithms with or without ER memory.\footnote{This result is practically meaningful as small $B$ also helps reduce the energy used for sending the feedback from the ER to ETs reliably, thus improving the net energy harvested. In this paper, for the purpose of exposition, we assume the feedback energy is negligible and the feedback bits are received by all ETs without error.}  

Next, we compare the performance of the two algorithms under the same value of $B$ assuming fixed $N_\text{t}$.

\begin{proposition}  \label{Statement:ErrorBound_A2_3}
Given $N_\text{t}\geq 2^B$, the phase-error upper bound of (A2) given in \eqref{Eq:A2_Worst} is always smaller than that of (A1) given in \eqref{Eq:ErrorBound} for the same value of $B\geq 1$.  
\end{proposition}
\begin{proof}
It follows that
\begin{equation*}
\frac{1}{2^B} \l(\frac{1}{2^B +1}\r)^{\frac{N_\text{t}}{2^B}-1} < \frac{1}{2^B} \l(\frac{1}{2^B}\r)^{\frac{N_\text{t}}{2^B}-1} = \l(\frac{1}{2^B}\r)^{\frac{N_\text{t}}{2^B}}.
\end{equation*}
The proof is thus completed. 
\end{proof}

Propositions~\ref{Statement:ErrorBound_A2_3} shows that (A2) is in general more efficient than (A1) in phase adaptation, including their respective optimal designs with $B=1$ according to Propositions~\ref{Statement:ErrorBound_A1} and \ref{Statement:ErrorBound_A2_2}. This is expected since (A2) exploits the memory at the ER and uses one additional feedback bit compared to (A1) for the same $B$ value.

\subsection{Efficiency of Distributed EB with Sequential Training}
In this subsection, we investigate the efficiency of the proposed distributed EB protocol with sequential training presented in Section~\ref{Subsection:Sequential}, after the training over $(M-1)$ phase adaptation intervals (see Fig.~\ref{Fig:Protocol_Sequential}), based on the performance analysis of the phase adaptation algorithms used in each phase adaptation interval given in the previous subsection. First, we present a lower bound on the efficiency $\eta$ defined in \eqref{Eq:Efficiency} in the following proposition.

\begin{proposition} \label{Proposition:LowerBound}
\begin{itemize}
\item Using (A1) (i.e., $B$-bit feedback without memory) for transmit phase adaptation, the efficiency of the proposed distributed EB protocol with sequential training is lower-bounded by
\begin{equation} \label{Eq:EfficiencyLower}
\eta \geq \frac{1}{Q^\star}\l(\sum_{m=1}^M \beta_m + \sum_{i,j=1, i\neq j}^M \sqrt{\beta_i \beta_j}\cos^2\l(2^{-B\frac{N_\text{t}}{2^B}} \pi\r) \r).
\end{equation}

\item Using (A2) (i.e., $(B+1)$-bit feedback with memory) for transmit phase adaptation, the efficiency $\eta$ is lower-bounded by
\begin{equation} \label{Eq:EfficiencyLower_A2}
\eta \geq \frac{1}{Q^\star}\l(\sum_{m=1}^M \beta_m + \sum_{i,j=1, i\neq j}^M \sqrt{\beta_i \beta_j}\cos^2\l( \frac{\pi}{2^B} \l(\frac{1}{2^B +1}\r)^{\frac{N_\text{t}}{2^B}-1}\r) \r).
\end{equation}
\end{itemize}
\end{proposition}
\begin{proof}
See Appendix~\ref{Proof:LowerBound}.
\end{proof}

It can be observed from \eqref{Eq:EfficiencyLower} and \eqref{Eq:EfficiencyLower_A2} that with fixed $B$, as the number of total training slots per phase adaptation interval $N_\text{t}$ becomes large, both lower bounds approach one. In other words, for the ideal case of $e = 0$ after each phase adaptation interval, which is obtained with $N_\text{t}\rightarrow \infty$, the proposed protocol achieves the maximum harvested power by the optimal EB given in \eqref{Eq:P_r max}. 

Next, we analyze the required number of training slots per phase adaptation interval with (A1) or (A2) to achieve a given target efficiency, denoted by $0<\hat{\eta}\leq 1$. By re-arranging the terms in the inequality $\eta\geq \hat{\eta}$ and using \eqref{Eq:EfficiencyLower} and \eqref{Eq:EfficiencyLower_A2}, we obtain the following corollary. 

\begin{corollary}
\begin{itemize}
\item If the number of training slots $N_\text{t}$ per phase adaptation interval in (A1) satisfies 
\begin{equation} \label{Eq:WorstCaseBound}
N_\text{t} \geq  \frac{2^B}{B}\log_2\l(\frac{\pi}{\arccos\l(\sqrt{\hat{\eta} - \frac{(1-\hat{\eta})\sum_{m=1}^M \beta_m}{\sum_{i,j=1, i\neq j}^M \sqrt{\beta_i \beta_j}} }\r)}\r),
\end{equation}
then it holds that $\eta\geq\hat{\eta}$. 

\item Furthermore, if $N_\text{t}$ in (A2) satisfies
\begin{equation} \label{Eq:WorstCaseBound_A2}
N_\text{t} \geq  2^B + 2^B\log_{2^B+1}\l(\frac{\pi/2^B}{\arccos\l(\sqrt{\hat{\eta} - \frac{(1-\hat{\eta})\sum_{m=1}^M \beta_m}{\sum_{i,j=1, i\neq j}^M \sqrt{\beta_i \beta_j}} }\r)}\r),
\end{equation}
then it holds that $\eta\geq\hat{\eta}$. 

\end{itemize}
\end{corollary}

It is worth noting from \eqref{Eq:WorstCaseBound} that for the special case of $\beta_1 = \beta_2 = ... = \beta_K = \beta$ (i.e., all ETs have identical channel gains to ER), and with the most efficient design with $B=1$, we have $\sum_{m=1}^M \beta_m = M\beta$ and $\sum_{i,j=1, i\neq j}^M \sqrt{\beta_i \beta_j} = M(M-1)\beta$. As a result, \eqref{Eq:WorstCaseBound} becomes
\begin{equation} \label{Eq:WorstCaseBoundBeta}
N_\text{t} \geq  2 \log_2\l(\frac{\pi}{\arccos\l(\sqrt{ \frac{M\hat{\eta} - 1}{M-1}}\r)}\r) .
\end{equation}
For instance, if $M=5$, i.e., if there are five ETs, \eqref{Eq:WorstCaseBoundBeta} yields $N_\text{t}\geq 9.6188$ with $\hat{\eta} = 0.99$ and $N_\text{t}\geq 12.9462$ with $\hat{\eta} = 0.999$, respectively. Thus, each ET needs $N_\text{t}=10$ and $N_\text{t}=14$ training slots with (A1) per phase adaptation interval (i.e., $N=5$ and $N=7$ feedback windows per phase adaptation interval, respectively, given $B=1$) to ensure 99\% and 99.9\% of the optimal EB gain, respectively, by  using the proposed distributed EB with sequential training.

\section{Numerical Results} \label{Section:Simulation}

In this section, we evaluate the performance of the proposed channel training and distributed EB schemes by simulation. For the simulation, we set the transmit power of each ET as $P=1$ Watt (W), and the number of signal paths between each ET$_m$ and ER as $L_m = 1$ (which corresponds to the line of sight (LoS) environment). Moreover, the channel power gain $\beta_m$ is modeled by  path-loss only, given by $\beta_m = c_0 (r_m/r_0)^{-\delta}$, where $c_0 = -20$ dB is a constant attenuation for the path-loss at a reference distance $r_0 = 1$ meter (m), $\delta = 3$ is the path-loss exponent, and $r_m$ is the distance between ET$_m$ and the ER. We assume the distance $r_m$ and the random phase shift $\theta_m$ of each ET$_m$ are distributed as  $r_m\sim \text{Uniform}(5,15)$ (in meters) and $\theta_m\sim\text{Uniform}(-\pi,\pi)$, respectively. 

\subsection{Comparison of Phase Adaptation Algorithms (A1) and (A2)}

\begin{figure}
\centering
\subfigure[(A1): $B$-bit feedback without memory]{
\centering
\includegraphics[width=8cm]{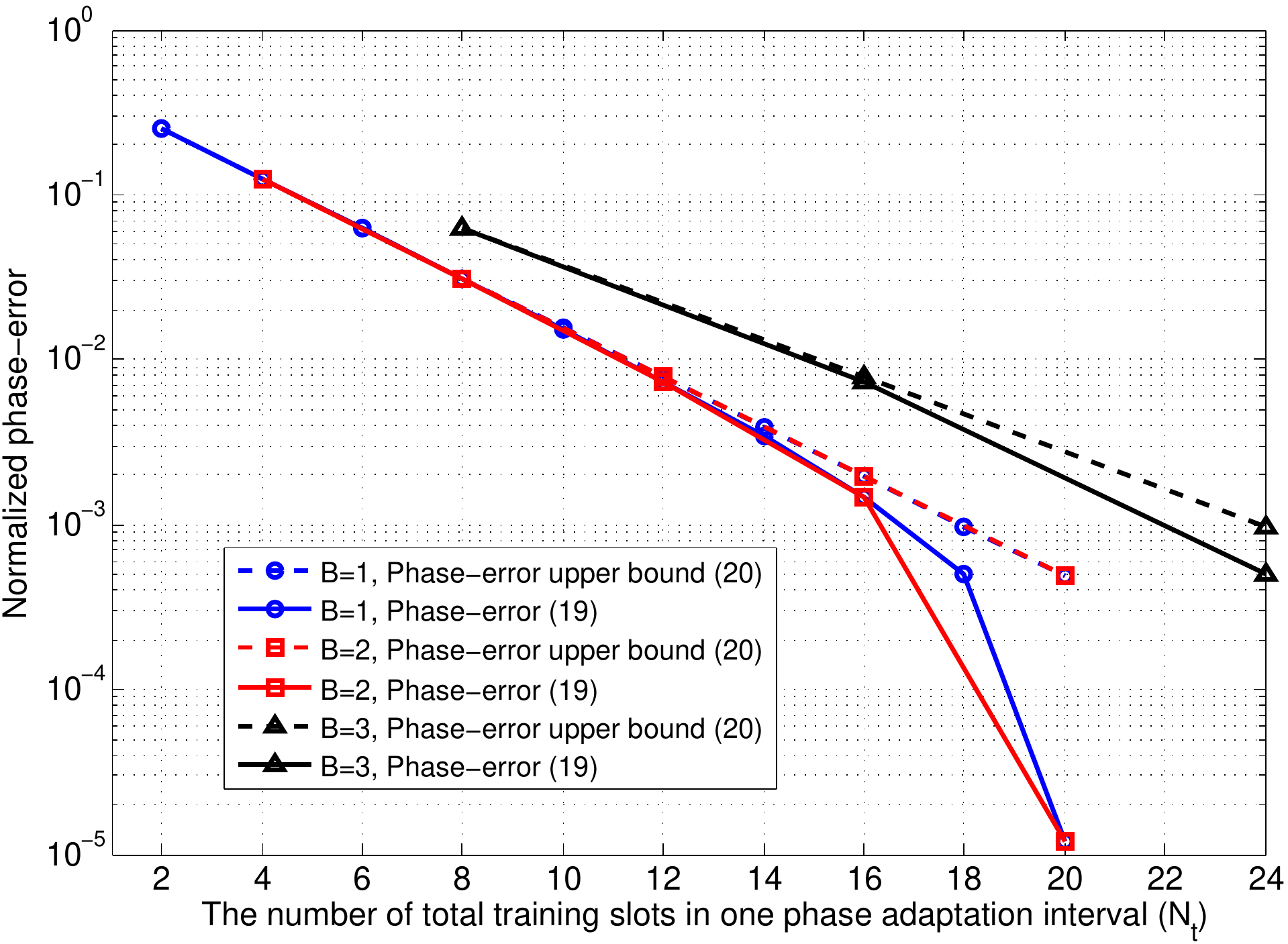}\label{Fig:PhaseError_A1}} 
\subfigure[(A2): $(B+1)$-bit feedback without memory]{
\centering
\includegraphics[width=8cm]{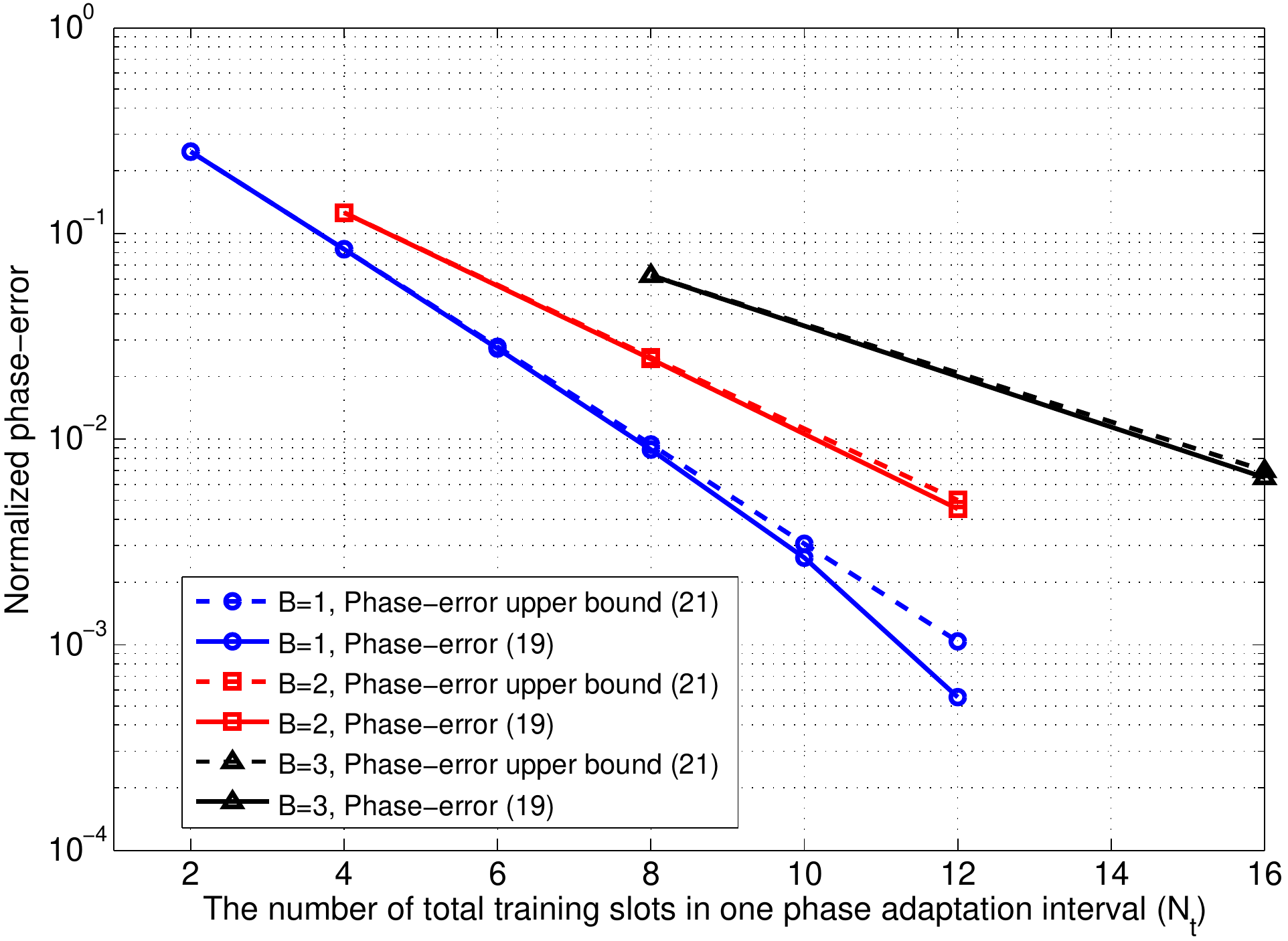}\label{Fig:PhaseError_A2}} 
\subfigure[Comparison of (A1) and (A2) with $B=1$]{
\centering
\includegraphics[width=8cm]{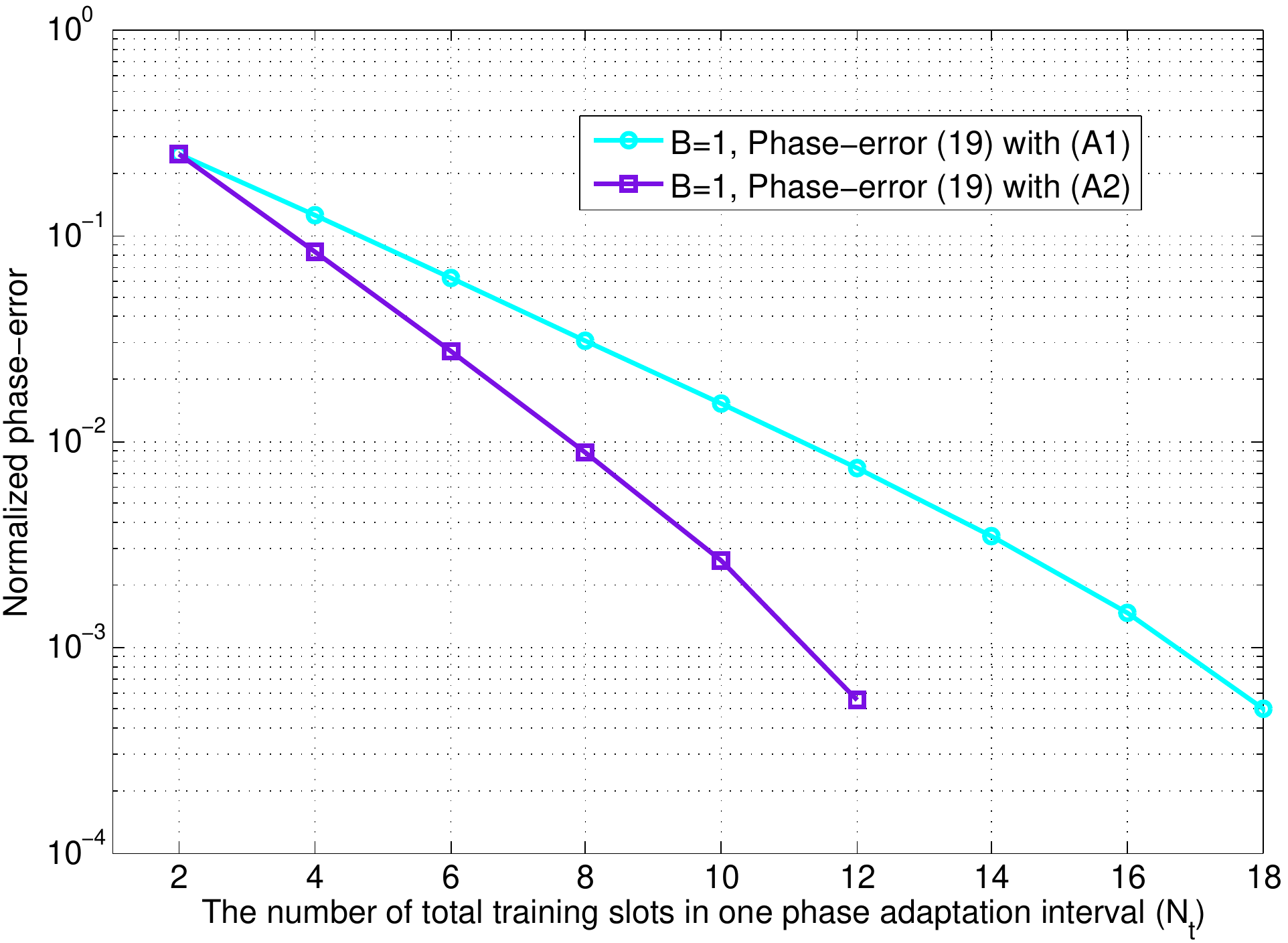}\label{Fig:PhaseError_A1A2}} 
\caption{Normalized phase-error $|e|/2\pi$ versus the total number of training slots $N_\text{t} = N 2^B$ in one phase adaptation interval. }
\label{Fig:PhaseError}
\end{figure}

We first compare the convergence performance of the two proposed phase adaptation algorithms, namely $B$-bit feedback without memory, i.e., (A1), and $(B+1)$-bit feedback with memory, i.e., (A2), presented in Sections~\ref{Subsection:A1} and \ref{Subsection:A2}, respectively. For this simulation, we set the number of ETs to be $M=5$, and consider the case of parallel training for one particular phase adaptation interval (similar results can be obtained for the case of sequential training and are thus omitted due to space limitation). We assume the index set of adapting ETs and non-adapting ETs to be $\cM_\text{A} = \{1,2\}$ and $\cM_\text{NA} = \{3,4,5\}$ (as a result, we have $\cM_\text{I} = \emptyset$), respectively. We assume a random realization of each $\theta_m$, $m=1,...,5$, and set the fixed phase of the non-adapting ETs to be $\bar{\phi}_3 = \bar{\phi}_4 = \bar{\phi}_5 = 0$.

Based on the above setup, in Fig.~\ref{Fig:PhaseError} we plot the normalized phase-error between the phase determined by the proposed (A1) or (A2) and the optimal phase, i.e., $|e|/2\pi$, where $e = \bar{\phi}_m - \psi^*$ as given in \eqref{Eq:Error}, versus the total number of training slots $N_\text{t}$ in one particular phase adaptation interval (see Fig.~\ref{Fig:TimeSlots}). The upper bounds of the phase-error for (A1) and (A2), given in \eqref{Eq:ErrorBound} and \eqref{Eq:A2_Worst}, respectively, are also plotted in Figs.~\ref{Fig:PhaseError_A1} and \ref{Fig:PhaseError_A2}, respectively. First, in Fig.~\ref{Fig:PhaseError_A1}, it is observed that the normalized phase-errors of (A1) with $B=1$ and $B=2$ are the same, which are smaller than that with $B=3$. This result is consistent with Proposition~\ref{Statement:ErrorBound_A1}. Next, in Fig.~\ref{Fig:PhaseError_A2}, it is observed that the smaller the value of $B$ is, the better is the normalized phase-error performance for (A2), which is in accordance with the result in Proposition~\ref{Statement:ErrorBound_A2_2}. Last, we compare the normalized phase-errors of (A1) and (A2) both with $B=1$ in Fig.~\ref{Fig:PhaseError_A1A2}. It can be seen that the performance of  (A2) is better than that of (A1), as expected from Proposition~\ref{Statement:ErrorBound_A2_3}.

In the following two subsections, we evaluate the performances of the proposed distributed EB protocols with sequential training and parallel training, respectively. Since $B=1$ is optimal for both (A1) and (A2) as shown in Section~\ref{Section:PhaseAdaptation} and verified by simulation in this subsection, we  consider $B=1$ in the rest of this section. As a result, the number of feedback windows per phase adaptation interval is given by $N = N_\text{t}/2$ (see Fig.~\ref{Fig:TimeSlots}).

\subsection{Distributed EB with Sequential Training}

\begin{figure} 
\centering
\includegraphics[width=9cm]{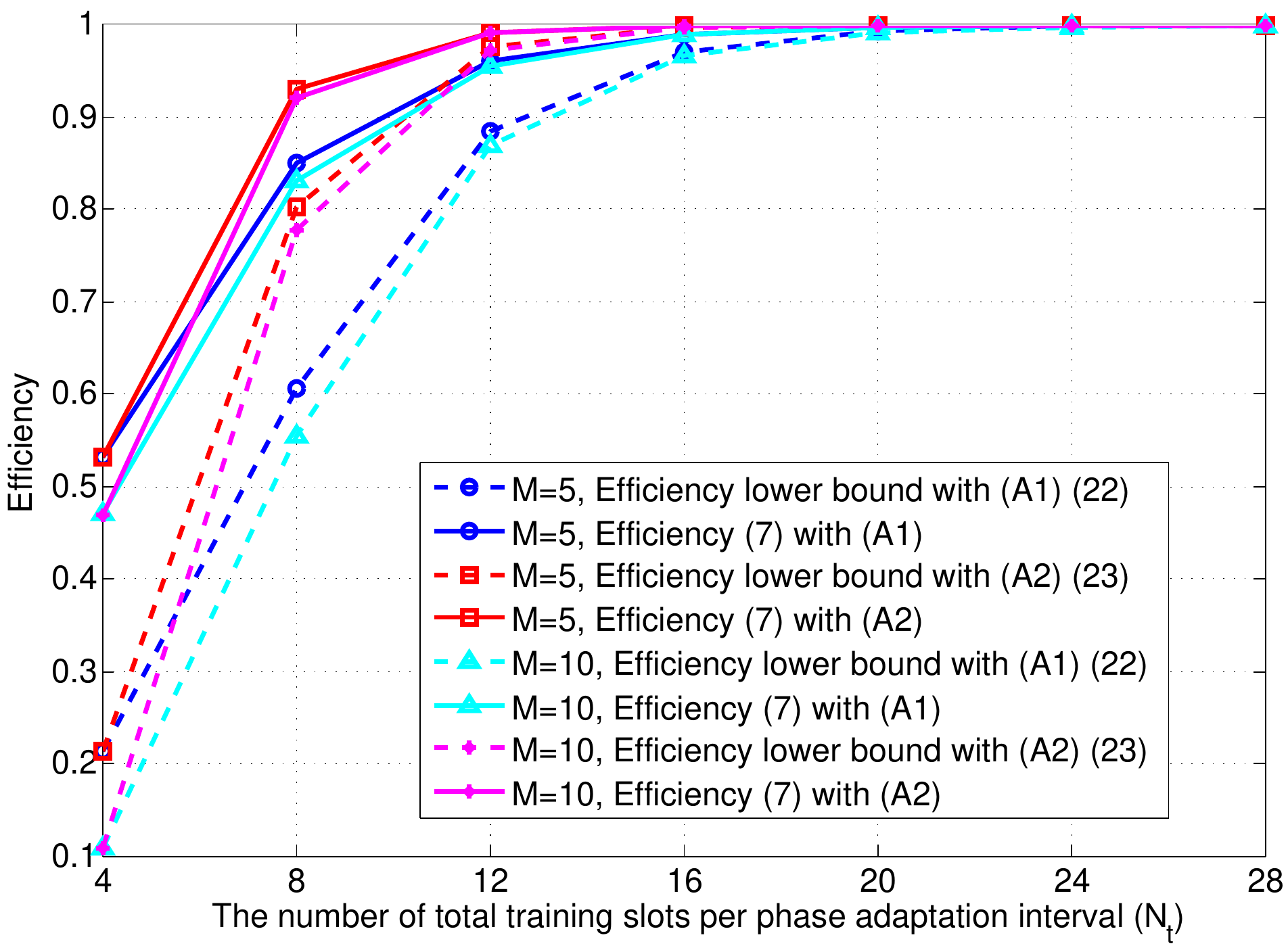}
\caption{The efficiency $\eta$ of the proposed distributed EB protocol with sequential training versus the number of training slots $N_\text{t}$ per phase adaptation intervals. } \label{Fig:Efficiency}
\end{figure}

Fig.~\ref{Fig:Efficiency} shows the distributed EB efficiency $\eta$ defined in \eqref{Eq:Efficiency}, with sequential training proposed in Section~\ref{Subsection:Sequential}, and its lower bounds with (A1) and (A2) given in \eqref{Eq:EfficiencyLower} and \eqref{Eq:EfficiencyLower_A2}, respectively, for the cases of $M=5$ and $M=10$, by averaging over 5000 randomly generated $r_m$ and $\theta_m$, $m=1,...,M$. First, it is observed that for both cases of $M=5$ and $M=10$,  the efficiency lower bounds given in \eqref{Eq:EfficiencyLower} and \eqref{Eq:EfficiencyLower_A2} become tighter as the number of total training slots $N_\text{t}$ per phase adaptation interval increases, and both eventually converge to 1, as compared to the exact efficiency given in \eqref{Eq:Efficiency}. However, the convergence is slightly faster for the case of $M=5$ than $M=10$. Furthermore, it can be seen from Fig.~\ref{Fig:Efficiency} that $N_\text{t} = 16$ already results in the efficiency higher than 0.95 for all cases plotted. Finally, it is observed that for both cases with $M=5$ and $M=10$, the efficiency with (A2) is higher than that with (A1) under each given $N_\text{t}$. 

\begin{figure} 
\centering
\includegraphics[width=9cm]{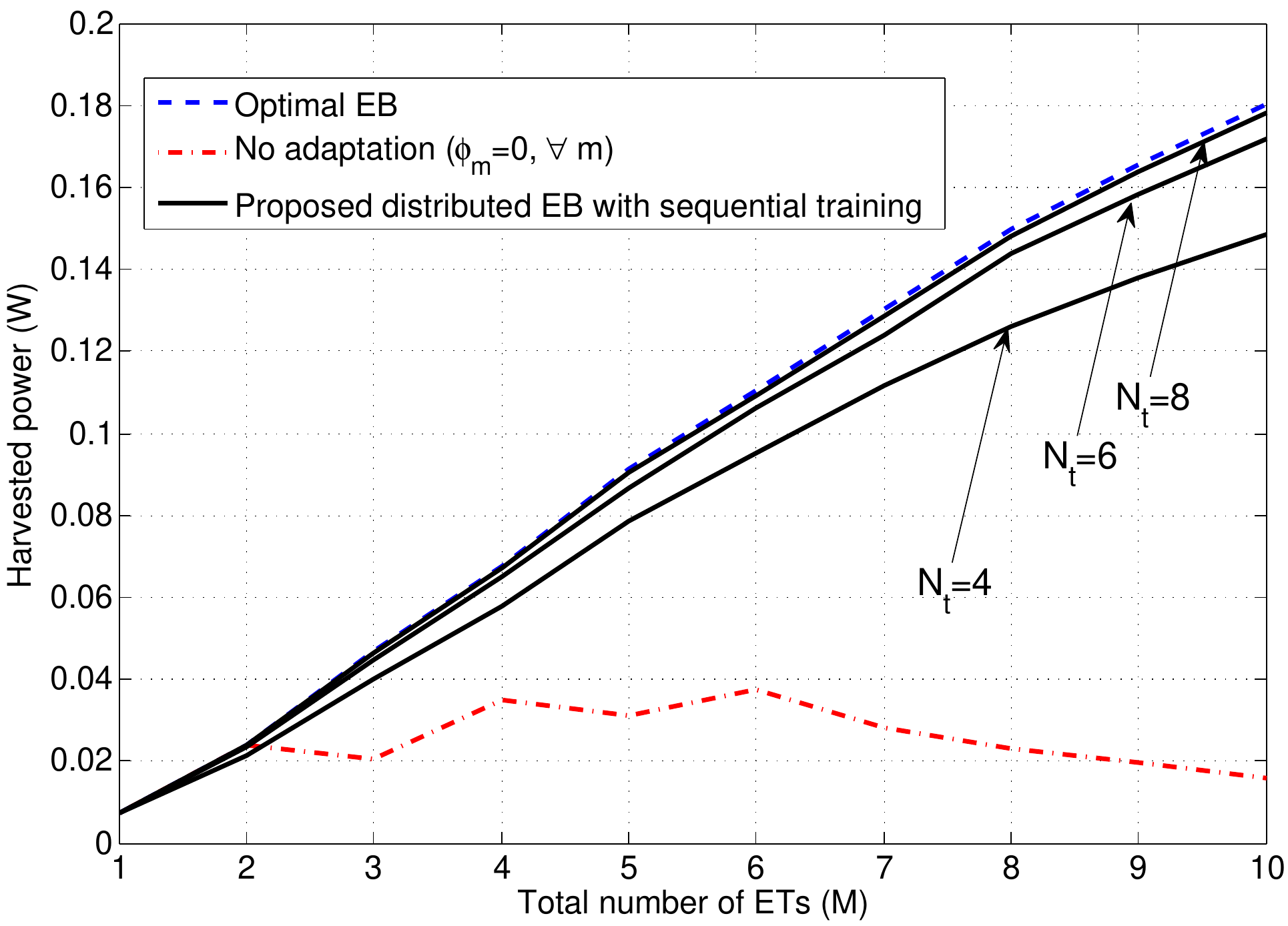}
\caption{The harvested power versus the total number of ETs, $M$. } \label{Fig:EBgain}
\end{figure}

In Fig.~\ref{Fig:EBgain}, we compare the harvested power by the proposed distributed EB protocol with sequential training versus the total number of ETs $M$ with different numbers of training slots $N_\text{t}$ per phase adaptation interval, for one set of random realizations of $r_m$ and $\theta_m$, $m=1,...,M$. Note that for the distributed EB protocol, we assume the use of $(B+1)$-bit feedback with memory, i.e., (A2), in Fig.~\ref{Fig:EBgain} since it has better performance than (A1) without memory. First, it is observed that the harvested power of the proposed distributed EB scheme keeps increasing with $M$, due to more significant EB gains; whereas that of no adaptation, in which each ET$_m$ fixes its phase to be $\bar{\phi}_m = 0$ at all time, fluctuates over $M$ in general. This is due to the fact that the channel phases from ETs to ER are different and as a result, their received signals may add constructively or destructively at ER. Second, it can be seen from Fig,~\ref{Fig:EBgain} that the larger the number of training slots $N_\text{t}$ per phase adaptation interval is, the higher is the harvested power achieved by the proposed distributed EB scheme, which is expected since larger $N_\text{t}$ yields more accurately estimated transmit phase for each ET$_m$ in (A2).

\begin{figure} 
\centering
\includegraphics[width=9cm]{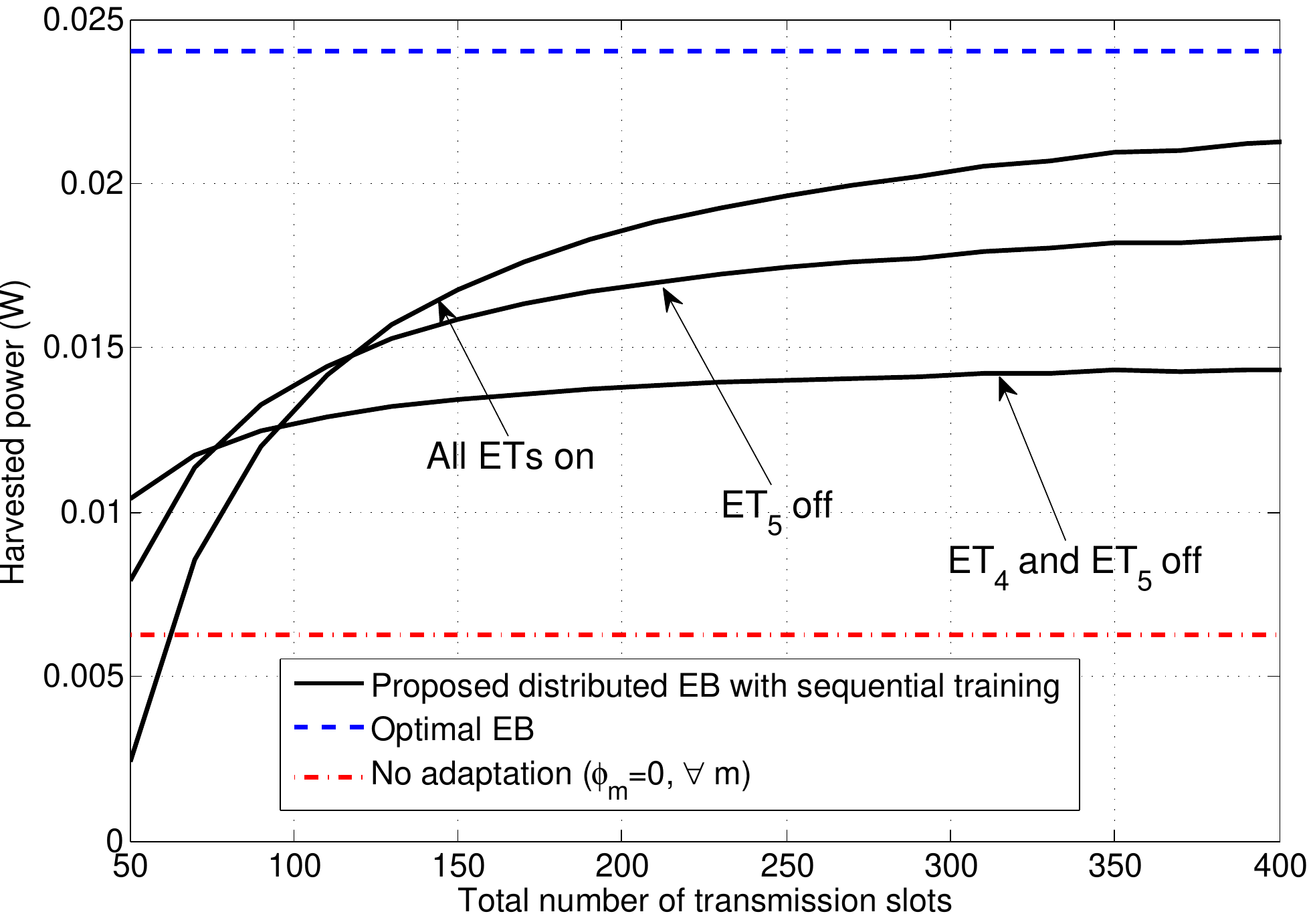}
\caption{The average harvested power versus the total number of transmission slots (including both training and energy transmission) with $M=5$ and $N=5$. } \label{Fig:TradeOff}
\end{figure}

In Fig.~\ref{Fig:TradeOff}, we plot the average harvested power by the proposed distributed EB with sequential training versus the total number of transmission slots (including both training and energy transmission shown in Fig.~\ref{Fig:Protocol_Sequential}), by averaging over 50000 randomly generated $r_m$ and $\theta_m$, $m=1,...,M$, to investigate the effect of the training overhead on the performance of the proposed distributed EB scheme. Note that we assume the use of (A2) with $B=1$ for the distributed EB protocol, $M=5$ and $N=5$. Moreover, by reordering the channel power gains of ETs as $\beta_1\geq\beta_2\geq...\geq\beta_5$, we also plot the average harvested powers for the cases without ET$_5$ (i.e., the ET with the weakest channel to ER is off) or both ET$_5$ and ET$_4$ (i.e., the two ETs with smallest channel gains are off).  It can be seen from Fig.~\ref{Fig:TradeOff} that if the total number of training slots is less than 120, it is no more optimal to let all ETs transmit as the ETs with weak channels do not contribute much to the overall EB gain but require the same training time for phase adaptation (otherwise, their received signals may not add coherently to other ETs' signals with stronger channels at ER); thus, they should be switched off to maximize the average harvested power at ER.  Moreover, when the total number of training slots is less than 60, the performance with all ETs on is even worse than that of no adaptation with $\phi_m=0$, $m=1,...,M$. However, as the total number of training slots increases, the average harvested power also increases with more ETs switched on and finally approaches the maximum harvested power by the optimal EB with all ETs on, due to the reduced training overhead. Notice that in practice, the total transmission time is constrained by the channel coherence time in any given propagation environment.  

\subsection{Distributed EB with Parallel Training}

\begin{figure} 
\centering
\includegraphics[width=9cm]{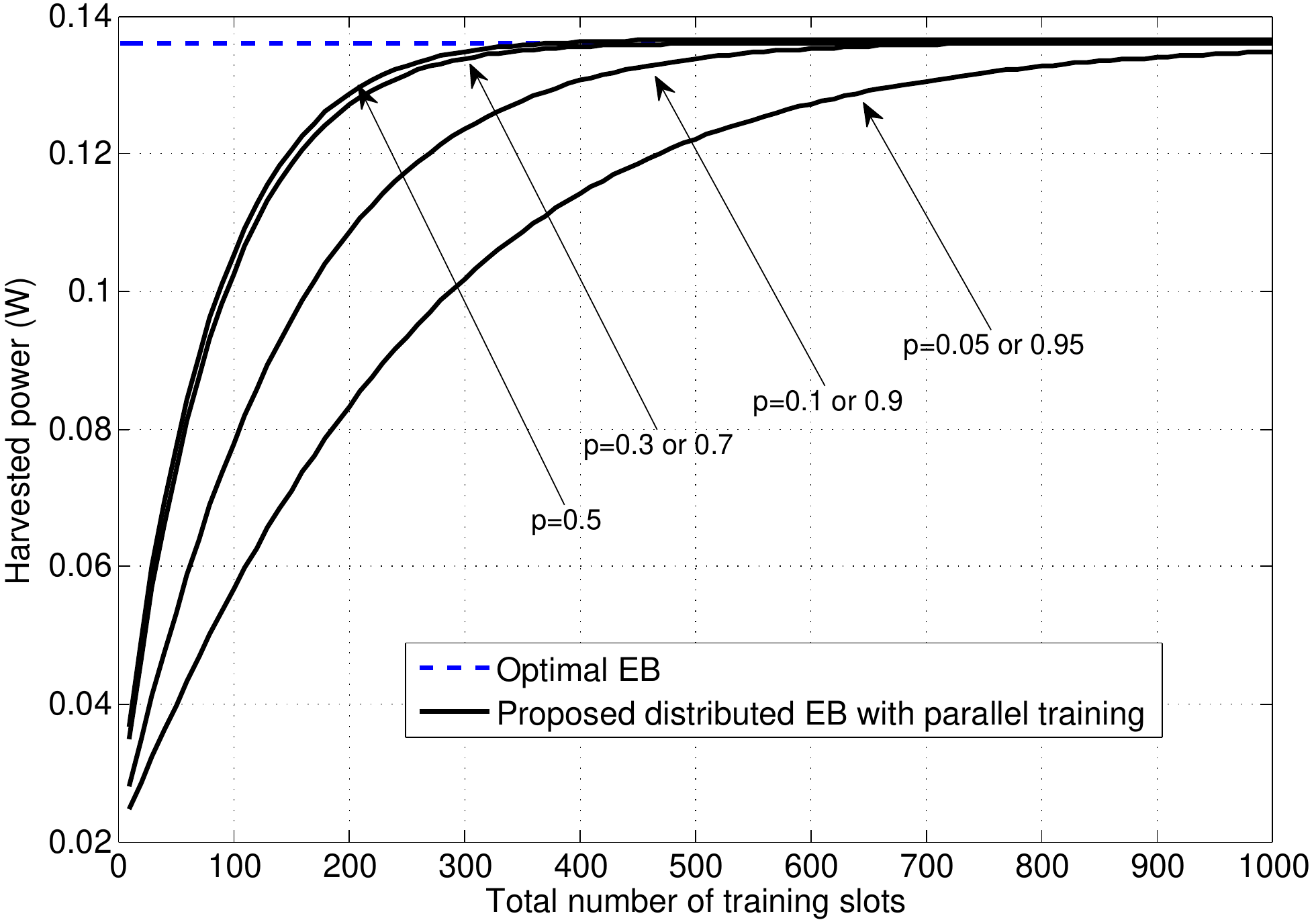}
\caption{The average harvested power of the proposed distributed EB protocol with parallel training versus the total number of training slots with different values of $p$.} \label{Fig:Coin}
\end{figure}

\begin{figure} 
\centering
\includegraphics[width=9cm]{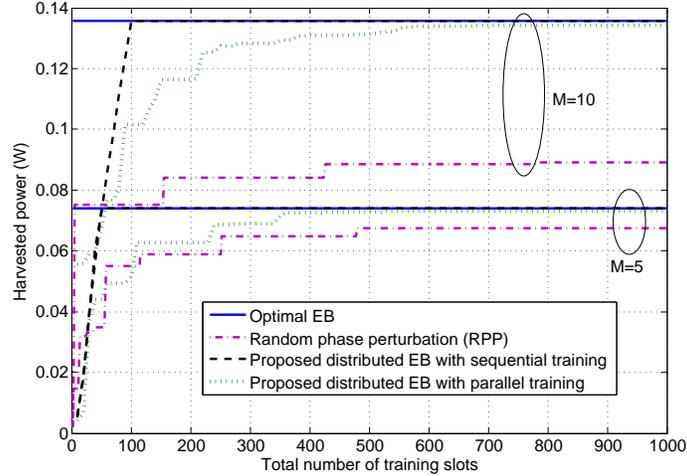}
\caption{The harvested power versus the total number of training slots.} \label{Fig:Comparison}
\end{figure}

Fig.~\ref{Fig:Coin} shows the convergence performance of our proposed distributed EB protocol with parallel training presented in Section~\ref{Subsection:parallel}, assuming the use of $(B+1)$-bit feedback with memory, i.e., (A2), for transmit phase adaptation. For comparison, we plot the results with different values of $p$ (i.e., the probability that each ET adapts its phase in each phase adaptation interval). The harvested power is averaged over 50000 randomly generated $r_m$ and $\theta_m$, $m=1,...,M$ in each feedback window. Moreover, we set $M=7$ and $N=5$ for the simulation. It can be seen from Fig.~\ref{Fig:Coin} that our proposed protocol with parallel training indeed converges to the maximum harvested power by optimal EB. Furthermore, it is observed that the convergence speed with $p=0.5$ is the fastest, and it gets slower as $p$ becomes more biased. The reason can be intuitively explained as follows. If  $p$ is too large, it is likely that the number of adapting ETs in $\cM_\text{A}$ becomes large, and as a result the gain from these ETs' phase adaptations to match the small (if any) number of non-adapting ETs is also small, rendering the overall convergence slower. A similar explanation also applies to the case with too small $p$.

In Fig.~\ref{Fig:Comparison}, we compare the convergence performance of the distributed EB protocols with sequential training versus parallel training, using (A2), for one set of random realizations of $r_m$ and $\theta_m$, $m=1,...,M$. We also compare our proposed protocols with an existing training scheme for distributed beamforming proposed in wireless communication \cite{J_MHMB:2010}, referred to as random phase perturbation (RPP), which is also applicable to distributed EB for WET of our interest. In this scheme, transmitters independently adjust their transmit phases via random perturbation based on one-bit feedback from the receiver, which indicates the increase or decrease of the current signal-to-noise-ratio (SNR) as compared to its recorded highest SNR. For this simulation, we set the total number of training slots per phase adaptation interval to be $N_\text{t} = 10$, and we present two cases with different number of ETs, i.e., $M=5$ and $M=10$, respectively. As a result, the distributed EB protocol with sequential training requires $N_\text{t}(M-1)= 40$ training slots in total for the case of $M=5$ and $90$ training slots for the case of $M=10$, in order for all ETs to set their phases $\bar{\phi}_m$, $m=1,...,M$  (see Figs.~\ref{Fig:Protocol_Sequential} and \ref{Fig:TimeSlots}). First, it is observed that the convergence speed of the parallel training is slower than that of the sequential training, which is expected due to the random selection of the adapting ETs in each phase adaptation interval. Next, it can be seen that the harvested power of our proposed distributed EB with sequential or parallel training converges much faster than that of the RPP benchmark scheme, and it is also much larger after convergence, especially when $M$ is large.

\section{Conclusion} \label{Section:Conclusion}
In this paper, we propose new channel training designs for distributed EB in WET systems, where the ETs adjust their transmit phases independently to achieve collaborative WET to a single ER. Based on a new phase adaptation algorithm for adapting ETs to  adapt their transmit phases to match that of other non-adapting ETs based on energy feedback from the ER with or without memory, we devise two distributed EB protocols with sequential training and parallel training with and without the need of centralized scheduling coordination, respectively. It is shown that smaller number of feedback bits per feedback window yields better convergence performance given the total training time for the proposed schemes. The proposed schemes are shown to converge to the optimal EB performance efficiently via both analysis and simulation, and also outperform the existing scheme based on random phase perturbation in terms of both convergence speed and energy efficiency. Possible future extensions of this work to the more general setup with multiple ERs in single- or multi-cell scenarios will be worth pursuing.  

\appendices 

\section{Proof of Lemma~\ref{Statement:ErrorBound_A2_1}} \label{Proof:ErrorBound_A2_1}

The phase-error upper bound in (A2) can be obtained by considering the worst-case scenario in (A2) which may lead to the largest phase-error. To this end, consider any given pair of two adjacent feedback windows, i.e., $n=k$ and $n=k+1$, respectively, with $k\geq 1$, for which we have the following four cases which are different combinations of Cases 1 and 2.
\begin{itemize}
\item $\cA^{(k)}$ is Case 1 and $\cA^{(k+1)}$ is Case 1: The size of the working set is reduced by $\frac{1}{2^B + 1}$.

\item $\cA^{(k)}$ is Case 1 and $\cA^{(k+1)}$ is Case 2: The size of the working set is reduced by $\frac{1}{2(2^B + 1)}$. 

\item $\cA^{(k)}$ is Case 2 and $\cA^{(k+1)}$ is Case 1: The size of the working set is reduced by $\frac{1}{2^B}$.

\item $\cA^{(k)}$ is Case 2 and $\cA^{(k+1)}$ is Case 2: The size of the working set is reduced by $\frac{1}{2^{B + 1}}$.
\end{itemize}

With the above four cases, it can be easily shown that the worst-case  scenario is that Case 1 is repeated over $n\leq N$ which leads to the largest size of the resulting working set $\cA^{(N)}$ over all combinations of Cases 1 and 2, which is given by $2\pi \frac{1}{2^B} \l(\frac{1}{2^B+1}\r)^{\frac{N_\text{t}}{2^B}-1}$. The phase-error upper bound in \eqref{Eq:A2_Worst} can thus be obtained by dividing the size of $\cA^{(N)}$ by $2$, since  $\psi_\text{best}$ is chosen as the center of $\cA^{(N)}$. The proof of Lemma~\ref{Statement:ErrorBound_A2_1} is thus completed.

\section{Proof of Proposition~\ref{Statement:ErrorBound_A2_2}} \label{Proof:ErrorBound_A2_2}
Proposition~\ref{Statement:ErrorBound_A2_2} can be proved by showing that the RHS of \eqref{Eq:A2_Worst} with $B=b$, where $b\geq 1$ is an integer, is always smaller than that of \eqref{Eq:A2_Worst} with $B=b+1$, given $N_\text{t}\geq 2^{b+1}$. In other words, we need to show that the following inequality holds for $b\geq 1$.
\begin{equation} \label{Eq:b b+1}
\frac{1}{2^b} \l(\frac{1}{2^b +1}\r)^{\frac{N_\text{t}}{2^b}-1} < \frac{1}{2^{b+1}} \l(\frac{1}{2^{b+1}+1}\r)^{\frac{N_\text{t}}{2^{b+1}}-1}.
\end{equation}

We first prove that \eqref{Eq:b b+1} holds for $b=1$. Substituting $b=1$ into \eqref{Eq:b b+1} and dividing the left-hand side (LHS) by the RHS yields
\begin{align}
\frac{6}{5}\l(\frac{1}{3}\r)^{\frac{N_\text{t}}{2}} \l(\frac{1}{5}\r)^{-\frac{N_\text{t}}{4}} = \frac{6}{5}\l(\frac{\sqrt{5}}{3}\r)^{\frac{N_\text{t}}{2}} < \frac{6}{5}\l(\frac{\sqrt{5}}{3}\r)< 1,
\end{align}
in which we have used the assumption that $N_\text{t}\geq 2^{b+1} > 2^b$. This shows that \eqref{Eq:b b+1} is true for $b=1$.

Next, we show that \eqref{Eq:b b+1} holds for $b\geq 2$. It follows that
\begin{align}
\frac{1}{2^b} \l(\frac{1}{2^b +1}\r)^{\frac{N_\text{t}}{2^b}-1} &< \frac{1}{2^b} \l(\frac{1}{2^b}\r)^{\frac{N_\text{t}}{2^b}-1} \\
& = \l(\frac{1}{2^b}\r)^{\frac{N_\text{t}}{2^b}}\\
&\overset{(a)} <  \l(\frac{1}{2^{b+2}}\r)^{\frac{N_\text{t}}{2^{b+1}}} \\
& = \frac{1}{2^{b+2}}\l(\frac{1}{2^{b+2}}\r)^{\frac{N_\text{t}}{2^{b+1}}-1} \\
&\overset{(b)} < \frac{1}{2^{b+1}}\l(\frac{1}{2^{b+1}+1}\r)^{\frac{N_\text{t}}{2^{b+1}}-1},
\end{align}
where $(a)$ can be shown to hold for $b\geq 2$, and $(b)$ is true for $b\geq 1$. Thus, \eqref{Eq:b b+1} also holds for $b\geq 2$. By combining the two cases of $b=1$ and $b\geq 2$, the proof of Proposition~\ref{Statement:ErrorBound_A2_2} is thus completed.

\section{Proof of Proposition~\ref{Proposition:LowerBound}} \label{Proof:LowerBound}
For the purpose of exposition, we denote $e_m$ as the phase-error given in \eqref{Eq:Error} for the $(m-1)$th phase adaptation interval at ET$_m$, with $2\leq m\leq M$. To prove Proposition~\ref{Proposition:LowerBound}, we need to show that the following inequality holds.
\begin{align} \label{Eq:Induction}
\frac{Q_\text{d}}{P} &\geq \sum_{m=1}^M \beta_m + \sum_{i=1,j=1, i\neq j}^M \sqrt{\beta_i \beta_j} \cos(e_i)\cos(e_j).
\end{align}
For simplicity, we assume $P=1$ without loss of generality. We prove \eqref{Eq:Induction} via mathematical induction as follows. First, we show that \eqref{Eq:Induction} holds for $M=2$. For convenience, we denote $Q_\text{d}$ given in \eqref{Eq:P_r Protocol} for $M=k$ as $Q_\text{d}^{(k)}$. In the case of $M=2$, when ET$_2$ adapts its phase $\phi_2$ via (A1), only ET$_1$ is transmitting with fixed phase $\bar{\phi}_1 = 0$ based on the protocol described in Section~\ref{Subsection:Sequential}, and thus the estimated phase at ET$_2$ is given by $\bar{\phi}_2 = \theta_2 - \theta_1 + e_2$. Thus, the harvested power at ER is given by
\begin{equation}
Q_\text{d}^{(2)} = \beta_1 + \beta_2 + 2\sqrt{\beta_1\beta_2}\cos(e_2),
\end{equation}
which implies that  \eqref{Eq:Induction} holds for $M=2$ with equality. Next, we assume that \eqref{Eq:Induction} holds for $M=k$, $k\geq 2$, i.e., the following inequality is true:
\begin{equation} \label{Eq:InductionAssumption}
Q_\text{d}^{(k)} \geq \sum_{m=1}^k \beta_m + \sum_{i=1,j=1, i\neq j}^k \sqrt{\beta_i \beta_j} \cos(e_i)\cos(e_j).
\end{equation}
Then, when $M=k+1$, the harvested power at ER by the distributed EB protocol is given by
\begin{align*}
Q_\text{d}^{(k+1)} &=  \beta_{k+1} + Q_\text{d}^{(k)} + 2 \cos(e_{k+1})\sqrt{\beta_{k+1}Q_\text{d}^{(k)}} \\
& \overset{(a)}\geq \beta_{k+1 } + \sum_{m=1}^k \beta_m + \sum_{i=1,j=1, k\neq l}^k \sqrt{\beta_i \beta_j} \cos(e_i)\cos(e_j)  \\
& \qquad + 2\cos(e_{k+1})\sqrt{\beta_{k+1}}\bigg(\sum_{m=1}^k \beta_m + \sum_{i=1,j=1, i\neq j}^k \sqrt{\beta_i \beta_j} \cos(e_i)\cos(e_j)\bigg)^{\frac{1}{2}}  \\
& \overset{(b)}\geq \sum_{m=1}^{k+1} \beta_m +  \sum_{i=1,j=1, i\neq j}^k \sqrt{\beta_i \beta_j} \cos(e_i)\cos(e_j) \\ 
&\qquad + 2\cos(e_{k+1}) \sqrt{\beta_{k+1}}\bigg(\sum_{m=1}^k \beta_m \cos^2(e_m)  +  \sum_{i=1,j=1, k\neq l}^k \sqrt{\beta_i \beta_j} \cos(e_i)\cos(e_j) \bigg)^{\frac{1}{2}} \\
& = \sum_{m=1}^{k+1} \beta_m +  \sum_{i=1,j=1, i\neq j}^{k+1} \sqrt{\beta_i \beta_j} \cos(e_i)\cos(e_j),
\end{align*}
where $(a)$ is due to the assumption in \eqref{Eq:InductionAssumption}, and $(b)$ is due to the fact that $0\leq \cos^2 (e_m) \leq 1$, $m=1,...,k$.  To summarize, we have shown that \eqref{Eq:Induction} holds for $M=k+1$ under the assumption that it holds for $M=k$. Since we have already shown that \eqref{Eq:Induction} is true for $M=2$, we conclude that \eqref{Eq:Induction} holds for any $M\geq 2$.  

Next, to prove \eqref{Eq:EfficiencyLower} and \eqref{Eq:EfficiencyLower_A2}, we substitute the error bounds given in \eqref{Eq:ErrorBound} and \eqref{Eq:A2_Worst} into \eqref{Eq:Induction}, respectively. It thus follows that
\begin{align} 
Q_\text{d}  \geq   \sum_{m=1}^M \beta_m + \sum_{i=1,j=1, i\neq j}^M \sqrt{\beta_i \beta_j} \cos^2\l(2^{-B\frac{N_\text{t}}{2^B}} \pi\r), \label{Eq:EfficiencyLowerProof}
\end{align}
\begin{align} 
Q_\text{d}  \geq   \sum_{m=1}^M \beta_m + \sum_{i=1,j=1, i\neq j}^M \sqrt{\beta_i \beta_j} \cos^2\l( \frac{\pi}{2^B} \l(\frac{1}{2^B +1}\r)^{\frac{N_\text{t}}{2^B}-1}\r), \label{Eq:EfficiencyLowerProof_A2}
\end{align}
respectively, since $\cos (t)$ is a decreasing function over $0\leq t \leq \pi$. Finally, the desired results in \eqref{Eq:EfficiencyLower} and \eqref{Eq:EfficiencyLower_A2} can be obtained by combining \eqref{Eq:EfficiencyLowerProof} and \eqref{Eq:EfficiencyLowerProof_A2} with \eqref{Eq:Efficiency}, respectively. The proof of Proposition~\ref{Proposition:LowerBound} is thus completed.

\bibliographystyle{IEEEtran}

\newpage

\end{document}